\pdfoutput=1
\documentclass[11pt]{article}
\usepackage[letterpaper,margin=1in]{geometry}
\usepackage{microtype}
\usepackage{amsmath,amssymb,amsfonts}
\usepackage{amsthm}
\usepackage{graphicx}
\usepackage{textcomp}
\usepackage{xcolor}
\usepackage{tikz}
\usetikzlibrary{arrows,positioning}
\usepackage{float}
\usepackage{csquotes}
\usepackage{enumitem}
\usepackage{prettyref}
\usepackage{bbm, bm}
\usepackage{xspace}
\usepackage[hidelinks]{hyperref}
\hypersetup{
  pdftitle={An Ω((log n/log log n)²) Cell-Probe Lower Bound for Dynamic Boolean Data Structures},
  pdfauthor={Young Kun Ko},
  pdfsubject={Full version of a paper appearing in FOCS 2026},
  pdfkeywords={cell-probe model, dynamic data structures, lower bounds, communication complexity, chronogram method}
}

\newtheorem{theorem}{Theorem}[section]

\newtheorem{lemma}[theorem]{Lemma}
\newtheorem{corollary}[theorem]{Corollary}

\newtheorem{claim}[theorem]{Claim}
\newtheorem{proposition}[theorem]{Proposition}
\newtheorem{definition}[theorem]{Definition}
\newtheorem{remark}[theorem]{Remark}

\newtheorem{fact}[theorem]{Fact}

\makeatletter
\newtheorem*{rep@theorem}{\rep@title}
\newcommand{\newreptheorem}[2]{%
\newenvironment{rep#1}[1]{%
\def\rep@title{#2 \ref{##1}}%
\begin{rep@theorem}}%
{\end{rep@theorem}}}
\makeatother

\newreptheorem{theorem}{Theorem}
\newreptheorem{lemma}{Lemma}
\newreptheorem{claim}{Claim}
\newreptheorem{definition}{Definition}

{\hspace*{\fill}$\Box$\par}
\newenvironment{proofof}[1]{\smallskip\noindent{\bf Proof of #1.}}%
{\hspace*{\fill}$\Box$\par}

\newcommand{\pref}{\prettyref}
\newrefformat{lem}{Lemma \ref{#1}}
\newrefformat{cl}{Claim \ref{#1}}
\newrefformat{prop}{Proposition \ref{#1}}
\newrefformat{prob}{Problem \ref{#1}}
\newrefformat{thm}{Theorem \ref{#1}}
\newrefformat{cor}{Corollary \ref{#1}}
\newrefformat{cha}{Chapter \ref{#1}}
\newrefformat{sec}{Section \ref{#1}}
\newrefformat{app}{Appendix \ref{#1}}
\newrefformat{tab}{Table \ref{#1}}
\newrefformat{fig}{Figure~\ref{#1}}
\newrefformat{conj}{Conjecture \ref{#1}}
\newrefformat{prot}{Protocol \ref{#1}}
\newrefformat{fact}{Fact \ref{#1}}
\newrefformat{def}{Definition \ref{#1}}
\newrefformat{rem}{Remark \ref{#1}}

\newcommand{\E}{{\mathbb{E}}}
\newcommand{\eps}{\varepsilon}

\newcommand{\cA}{\mathcal{A}}

\newcommand{\cD}{\mathcal{D}}
\newcommand{\cE}{\mathcal{E}}

\newcommand{\cM}{\mathcal{M}}

\newcommand{\cS}{\mathcal{S}}
\newcommand{\cQ}{\mathcal{Q}}
\newcommand{\cU}{\mathcal{U}}
\newcommand{\cT}{\mathcal{T}}

\newcommand{\Patrascu}{P\u{a}tra\c{s}cu}

\newcommand{\ip}[1]{\left\langle #1 \right\rangle}

\newcommand{\Field}{\mathbb{F}}

\newcommand{\abs}[1]{\left| #1 \right|}
\newcommand{\norm}[1]{\| #1 \|}

\newcommand{\KL}[2]{D ( {#1} || {#2} )}

\newcommand{\poly}{{\operatorname{poly}\xspace}}

\newcommand{\Ber}{{\mathsf{Ber}\xspace}}

\newcommand{\adv}{\mathsf{adv}}

\newcommand{\EC}[1]{}

\DeclareMathOperator*{\argmax}{arg\,max}

\floatstyle{boxed}
\newfloat{Protocol}{H}{pro}

\floatstyle{boxed}
\newfloat{Problem}{H}{pro}

\title{An $\Omega((\log n/\log\log n)^2)$ Cell-Probe Lower Bound \\
for Dynamic Boolean Data Structures\thanks{A preliminary version of this paper appears in the Proceedings of the IEEE Symposium on Foundations of Computer Science (FOCS 2026).}}

\author{Young Kun Ko\thanks{Department of Computer Science and Engineering, Pennsylvania State University, University Park, PA, USA. \texttt{ykko@psu.edu}}}

\date{}

\begin{document}
\maketitle

\begin{abstract}

We resolve the long-standing open problem of Boolean dynamic data structure hardness, proving an unconditional lower bound of $\Omega (( \log n / \log \log n)^2)$ for the Multiphase Problem of \Patrascu~[STOC 2010] (instantiated with Inner Product over $\mathbb{F}_2$). This matches the $\widetilde{\Omega}(\log^2 n)$ bound that Larsen [STOC 2012] established for {\em weighted} problems, where each query returns $\Omega ( \log n )$ bits, and closes the gap left by the $\widetilde{\Omega}(\log^{1.5} n)$ Boolean bound of Larsen, Weinstein, and Yu [STOC 2018]. It also answers, up to a $( \log \log n )^2$ factor, Open Question 15.24 in Viola's textbook {\em Mathematics of the Impossible} [Cambridge 2026].

The previous barrier was methodological: all prior works relied on ``one-way'' communication games, where the inability to verify query simulations necessitated complex machinery (such as the Peak-to-Average Lemma) that hit a hard ceiling at $\log^{1.5} n$. Our key contribution is conceptual: we introduce a 2.5-round Multiphase Communication Game that augments the standard one-way model with a verification round, where Bob confirms the consistency of Alice's simulation against the actual memory. This simple, qualitative change removes the obstruction and yields the full bound directly. As a consequence, our analysis naturally extends to other hard Boolean functions --- those of small discrepancy under product distributions with large min-entropy --- offering a general recipe for translating discrepancy lower bounds into $\Omega (( \log n / \log \log n)^2)$ dynamic Boolean data structure lower bounds.

We also argue that this result likely represents the structural ceiling of the Chronogram framework initiated by Fredman and Saks [STOC 1989]: any $\omega(\log^2 n)$ lower bound would require either fundamentally new techniques or major circuit complexity breakthroughs.

\end{abstract}

\newpage
\section{Introduction}

Proving unconditional lower bounds for dynamic data structures in the cell-probe
model \cite{yao_should_1981} stands as one of the grand challenges in theoretical computer
science. In this model, memory is organized into fixed-size cells of
$w = \Theta( \log n )$ bits, and we charge unit cost for accessing a cell while permitting
arbitrary computation on the probed data for free. Because of this unbounded
computational power, the cell-probe model strictly subsumes all realistic
architectures; consequently, any lower bound proved here applies universally.
However, this generality comes at a steep price: the model's ability to
arbitrarily compress and encode information makes proving lower bounds
notoriously difficult.

Despite more than 35 years of sustained effort since Fredman and Saks' seminal
Chronogram method \cite{fredman_cell_1989}, progress has been slow (see
\pref{fig:history}). The field advanced from $\Omega ( \log n / \log \log n )$ to
$\Omega ( \log n )$ over 15 years \cite{patrascu_logarithmic_2006}, and then, over
another 6 years, to $\Omega ( (\log n / \log \log n )^2 )$ \cite{larsen_cell_2012} ---
but only for {\em weighted} problems, where each query returns $\Omega ( \log n )$ bits.

For Boolean (decision) problems, where each query returns a single bit, the
$\Omega ( (\log n / \log \log n )^2 )$ barrier seemed insurmountable, and the question
was posed from two directions. \Patrascu~believed that dynamic $2$-dimensional range
parity --- report whether the number of points inside a query rectangle is odd ---
requires $\widetilde{\Omega} ( \log^2 n )$, having proved an
$\widetilde{\Omega} ( \log n )$ query lower bound for the static problem under
near-linear space \cite{patrascu_lower_2007}; Thorup lists this among the five open
problems \Patrascu~regarded as most important \cite{thorup_mihai_2013}.\footnote{Our
hard problem is the Multiphase Problem rather than range parity; an
$\Omega ( (\log n / \log \log n )^2 )$ bound for range parity itself remains open.}
Viola's textbook, which spans computational complexity as a whole, poses only twelve open questions; one of them~\cite[Open Question 15.24]{viola_mathematics_2026} asks for exactly this threshold:
``Are there explicit (say in FP) boolean functions $f : [2]^n \to [2]^m$ with $m$ a power
of $n$ that require time $\geq \log^2 n$ for dynamic data structures?''

The obstruction was that no technique reached this threshold at all. Larsen's
$\Omega ( (\log n / \log \log n )^2 )$ bound \cite{larsen_cell_2012} applies only to
the weighted variant, where each point carries an $O ( \log n )$-bit weight and the
query asks for the range sum. Such weights can encode more information than a memory
address, and Larsen's technique exploits this heavily; the argument does not survive
the passage to a single output bit. Larsen, Weinstein and Yu \cite{larsen_crossing_2020}
extended the method to the Boolean setting, but the Peak-to-Average Lemma they
introduced loses a $\sqrt{\log n}$ factor and stops at
$\widetilde{\Omega} ( \log^{1.5} n )$ --- a loss that is inherent to the lemma
\cite[Appendix B]{larsen_crossing_2020}.

\begin{figure}[!h]
\centering
\begin{tikzpicture}[scale=1.2]
  \draw[thick,->] (0,0) -- (12,0) node[right] {Year};

  \foreach \x/\year/\label in {
    1/1989/\cite{fredman_cell_1989},
    3/2006/\cite{patrascu_logarithmic_2006},
    5/2012/\cite{larsen_cell_2012},
    8/2018/\cite{larsen_crossing_2020},
    11/2026/Our Result
  } {
    \draw (\x,0.1) -- (\x,-0.1);
    \node[above] at (\x,0.1) {\year};
    \node[below] at (\x,-0.3) {\small \label};
  }

  \node[align=center] at (1,-1.2) {$\widetilde{\Omega}\left( \log n \right)$\\ \small  (Boolean)};
  \node[align=center] at (3,-1.2) {$\Omega(\log n)$\\\small  (Boolean)};
  \node[align=center] at (5,-1.2) {$\widetilde{\Omega}\left( \log^2 n \right)$\\ \small  \textcolor{red}{(weighted)}};
  \node[align=center] at (8,-1.2) {$\widetilde{\Omega}(\log^{1.5} n)$\\ \small (Boolean)};
  \node[align=center] at (11,-1.2) {$\widetilde{\Omega}\left( \log^2 n\right)$ \\ \small (Boolean)};
\end{tikzpicture}
\caption{Historical progression of dynamic cell-probe lower bounds. Throughout, $\widetilde{\Omega}$ hides a $\poly ( \log \log n )$ factor in the denominator.}
\label{fig:history}
\end{figure}
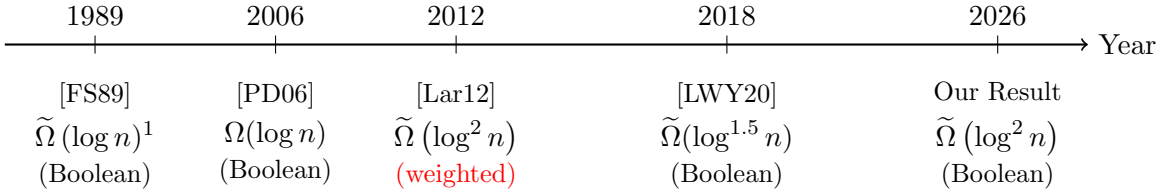

In this paper, we attain the threshold. We prove an unconditional lower bound of
$\Omega ( (\log n / \log \log n )^2 )$ for the Multiphase Problem with Inner Product
over $\Field_2$, the first such bound for a Boolean problem, matching the highest known
bound for weighted problems and answering, up to a $( \log \log n )^2$ factor, Viola's Open Question
15.24~\cite{viola_mathematics_2026}. Our result directly
implies lower bounds for fundamental problems including dynamic matrix-vector
multiplication over $\Field_2$, dynamic path parity, and the counting version of
Erickson's Problem (\pref{prob:erickson}).

\subsection{Our Contribution} \label{sec:contribution}

As our hard problem, we consider the Multiphase Problem of \Patrascu~\cite{patrascu_mihai_towards_2010}.
For a Boolean function $f: \{0,1\}^{2n} \to \{ \pm 1 \}$, the problem is:
\begin{Problem}
    \begin{itemize}
    \item $\vec{S} = S_1, \ldots , S_m \in \{ 0, 1 \}^n$ are pre-processed.
    \item The bits of $X \in \{ 0 , 1\}^n$ arrive one at a time, as a sequence of $n$ single-bit updates, each costing $t_u$ probes.
    \item Per any given $q \in [m]$, output $f ( S_q, X )$ using $t_{tot}$ total probes.
    \end{itemize}
    \caption{Multiphase Problem for $f$ \label{prob:Multiphase}}
\end{Problem}

We stress that a single update sets one bit of $X$, so that the vector $X$ specifies $n$ updates rather than a single one, at a total update cost of $t_u \cdot n$ probes. It is this sequence of $n$ single-bit updates that the Chronogram method partitions into epochs; see the \emph{Epochs} paragraph of \pref{sec:main_proof}.

When $f$ is Disjointness or Inner Product over $\Field_2$, \Patrascu~conjectured that $\max \{ t_u, t_{tot} \} \geq n^{\Omega(1)}$ for $m = n^{1+\Omega(1)}$~\cite{patrascu_mihai_towards_2010} --- the notorious Multiphase Conjecture \cite{patrascu_mihai_towards_2010, thorup_mihai_2013, braverman_communication_2022}. \Patrascu~showed that this conjecture implies polynomial lower bounds for many fundamental dynamic problems
including dynamic $s$-$t$ reachability.

We show the following main theorem for the Multiphase Problem.

\begin{theorem} \label{thm:main_informal}
    Let $f$ be Inner Product over $\Field_2$ and let $m = n^{1+\Omega(1)}$. Then any data structure for the Multiphase Problem for $f$ whose update time satisfies $t_u = n^{o(1)}$ must have
    $$ t_{tot} \geq \Omega \left( \left( \frac{\log n}{\log ( t_u w )} \right)^2 \right).$$
    In particular,
    $$ \max \{ t_u, t_{tot} \} \geq \Omega \left( \left( \frac{\log n}{\log \log n} \right)^2 \right). $$
\end{theorem}

We prove \pref{thm:main_informal} in \pref{sec:main_proof}.

\paragraph{General Lifting Theorem}
While \pref{thm:main_informal} focuses on the Inner Product over $\Field_2$ to resolve the specific open problem posed by \Patrascu~and Larsen, our technique is not limited to the Inner Product over $\Field_2$. Our Simulation Theorem (\pref{sec:simulation}) establishes a general translation from communication complexity to dynamic data structure lower bounds. Our lower bound applies to \emph{any} function $f$ that is ``hard.'' Roughly, $f$ is hard if no {\em product distribution} with large min-entropy allows an advantage greater than $n^{-\Omega(1)}$ (i.e., $f$ has low discrepancy under such distributions). This hardness condition is indeed satisfied by the Inner Product over $\Field_2$ function, but unfortunately, the Disjointness function does not meet the criterion \cite{braverman_information_2013}. A naive strategy of sampling and revealing a few coordinates of $X$ already yields a nontrivial advantage for Disjointness, violating the hardness condition. We refer the reader to \pref{app:lifting}, specifically \pref{def:hard}, for the precise technical definition. We then prove a general ``lifting'' theorem (\pref{thm:lifting}) showing that if $f$ is ``hard,'' then the Multiphase Problem with $f$ must have, for $m = n^{1+\Omega(1)}$ and $t_u = n^{o(1)}$,
\begin{equation*}
    t_{tot} \geq \Omega \left( \left( \frac{\log n}{\log ( t_u w )} \right)^2 \right).
\end{equation*}
We defer the full proof to \pref{app:lifting} to preserve the clarity of the main argument for \pref{thm:main_informal}.

\subsection{Applications}

Via the standard reductions of \Patrascu~\cite{patrascu_mihai_towards_2010}, this directly implies lower bounds for the following problems.

\begin{Problem}
    \begin{itemize}
        \item A matrix $M \in \Field_2^{m \times n}$ is given and pre-processed.
        \item $X \in \{ 0 , 1 \}^n$ is updated dynamically. Output $\ip{M_q, X}$ for any given $q \in [m]$.
    \end{itemize}
    \caption{Dynamic Matrix-Vector Multiplication over $\Field_2$ \label{prob:multiplication}}
\end{Problem}

\begin{Problem}
    \begin{itemize}
    \item A directed graph $G = (V,E)$ is initially given. Updates add and remove edges in the graph.
    \item For a given $(u,v) \in V \times V$, output the parity of the number of paths between $u$ and $v$. For the $k$-Path Parity Problem, output the parity of the number of paths of length exactly $k$.
    \end{itemize}
    \caption{Dynamic ($k$-) Path Parity \label{prob:path}}
\end{Problem}

\begin{Problem}
    \begin{itemize}
    \item A matrix of integers is preprocessed. Updates increment all values in a specified row or column.
    \item For a given query, output the parity of the number of maximum values in the matrix.
    \end{itemize}
    \caption{Erickson's Problem, Counting Version \label{prob:erickson}}
\end{Problem}

\begin{corollary}
    \pref{prob:multiplication}, \pref{prob:path} and \pref{prob:erickson} require $\max \{ t_u, t_{tot} \} \geq \Omega \left( \left( \frac{\log n}{\log \log n} \right)^2 \right)$.
\end{corollary}

\subsection{Previous Works}

To put our technical contribution into perspective, we briefly summarize the technical developments in dynamic cell-probe lower bounds. Readers familiar with these developments can skip to \pref{sec:technical} for our technical contribution.

\subsubsection{Explicit Quests for Poly-logarithmic Lower Bounds} \label{sec:polylog}

\paragraph{Chronogram Method}

Fredman and Saks \cite{fredman_cell_1989} introduced the Chronogram method. At a high level, the Chronogram method (\pref{fig:Chronogram}) divides the sequence of $n$ updates $X$ into $\ell$ epochs, $\{ X^{(i)} \}_{i=1}^\ell$, where $X^{(i)}$ consists of $n_i$ bits and $n_i = |X^{(i)}| = \beta^i$, where $\beta$ is a parameter set per problem in question (roughly $\poly ( w t_u )$). We then process the updates in reverse order, that is, the larger epochs are processed first, and we denote the updates in epoch $i$ by $U_i$.

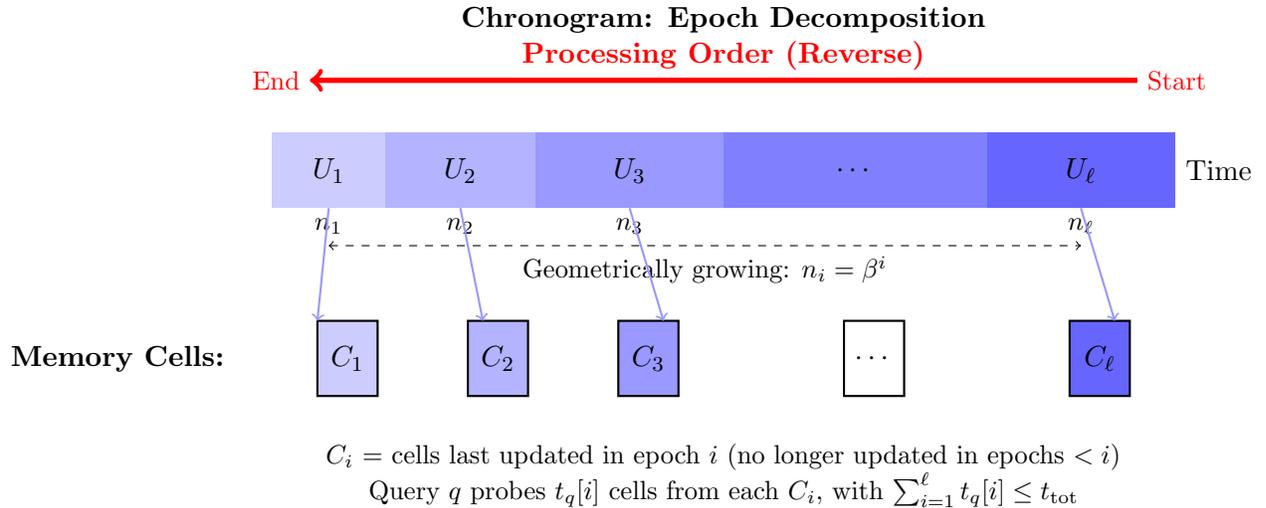
\begin{figure}[!h]
\centering
\resizebox{\textwidth}{!}{%
\begin{tikzpicture}[scale=1.0]
  \node[above] at (6, 5.2) {\textbf{Chronogram: Epoch Decomposition}};

  \draw[->, thick] (0,3.5) -- (12,3.5) node[right] {Time};

  \fill[blue!20] (0,3) rectangle (1.5,4);
  \fill[blue!30] (1.5,3) rectangle (3.5,4);
  \fill[blue!40] (3.5,3) rectangle (6,4);
  \fill[blue!50] (6,3) rectangle (9.5,4);
  \fill[blue!60] (9.5,3) rectangle (12,4);

  \node at (0.75,3.5) {$U_1$};
  \node at (2.5,3.5) {$U_2$};
  \node at (4.75,3.5) {$U_3$};
  \node at (7.75,3.5) {$\cdots$};
  \node at (10.75,3.5) {$U_\ell$};

  \node[below, font=\small] at (0.75,3) {$n_1$};
  \node[below, font=\small] at (2.5,3) {$n_2$};
  \node[below, font=\small] at (4.75,3) {$n_3$};
  \node[below, font=\small] at (10.75,3) {$n_\ell$};

  \draw[<->, dashed] (0.75,2.5) -- (10.75,2.5);
  \node[below, font=\small] at (5.75,2.5) {Geometrically growing: $n_i = \beta^i$};

  \draw[->, very thick, red, line width=2pt] (11.5, 4.7) -- (0.5, 4.7);
  \node[above, red, font=\bfseries] at (6, 4.7) {Processing Order (Reverse)};
  \node[right, red, font=\small] at (11.5, 4.7) {Start};
  \node[left, red, font=\small] at (0.5, 4.7) {End};

  \node[left] at (-0.5,1) {\textbf{Memory Cells:}};

  \foreach \x/\label/\color in {
    1/C_1/blue!20,
    3/C_2/blue!30,
    5/C_3/blue!40,
    8/\cdots/white,
    11/C_\ell/blue!60
  } {
    \draw[fill=\color, thick] (\x-0.4,0.5) rectangle (\x+0.4,1.5);
    \node at (\x,1) {$\label$};
  }

  \draw[->, thick, blue!40] (0.75,3) -- (0.6,1.5);
  \draw[->, thick, blue!40] (2.5,3) -- (2.8,1.5);
  \draw[->, thick, blue!40] (4.75,3) -- (5.2,1.5);
  \draw[->, thick, blue!40] (10.75,3) -- (11.2,1.5);

  \node[below, align=center, font=\small] at (6,0) {
    $C_i$ = cells last updated in epoch $i$ (no longer updated in epochs $< i$) \\
    Query $q$ probes $t_q[i]$ cells from each $C_i$, with $\sum_{i=1}^\ell t_q[i] \leq t_{tot}$
  };

\end{tikzpicture}%
}
\caption{
  The Chronogram technique decomposes $n$ updates into $\ell = \Theta(\log_{\beta} n)$ epochs with geometrically growing sizes $n_i = \beta^i$. Crucially, epochs are \textbf{processed in reverse order} (largest first). Each epoch $i$ has associated cells $C_i$ that were last updated in that epoch. A query algorithm for $q \in \cQ$ probes $t_q[i]$ cells from each $C_i$.
}
\label{fig:Chronogram}
\end{figure}
The main observation by Fredman and Saks is that we can decompose the updated cells depending on the latest epoch that performed the update. That is $C_1, \ldots, C_{\ell}$, where each $C_i \subseteq U_i$ contains the cells last updated by epoch $i$, and no longer updated in epochs $<i$. The key point is that there must exist $i \in [\ell]$ such that the query algorithm probes at most $O ( t_{tot} / \ell )$ cells from $C_i$ on average. That is $\E_{q \sim \cQ} \left[ t_q [i] \right] \leq O( t_{tot} / \ell)$. This was elegantly phrased in \cite{larsen_super-logarithmic_2025} as a ``static'' data structure with pre-initialized memory and a cache, where pre-initialized memory refers to $C_{>i}$ and cache refers to $C_{<i}$. The goal then is to show a lower bound on $t_{tot} / \ell$, which is the average number of probes into $C_i$. Fredman and Saks proved an $\Omega(1)$ lower bound on $t_{tot} / \ell$, leading to $\max \{ t_u, t_{tot} \} \geq \Omega ( \log n / \log \log n )$. It then took more than 15 years for further quantitative improvement. A more careful analysis by considering the information transfer between the epochs \cite{patrascu_logarithmic_2006} led to shaving the $\log \log n$ factor in the denominator.

\paragraph{Chronogram + Cell-Sampling}

The breakthrough work of Larsen \cite{larsen_cell_2012} then successfully merged the cell-sampling technique~\cite{siegel_universal_2004, panigrahy_lower_2010} with the Chronogram technique \cite{fredman_cell_1989} to give an $\Omega ( (\log n / \log \log n)^2 )$ bound with the following rough idea. Recall that after the decomposition using the Chronogram method, the goal is to show a lower bound on $t_{tot} / \ell$, which is the number of probes into $C_i$. As $\ell$ would always be roughly $ \Theta ( \log n / \log ( t_u w ) )$, if we can show, say, $t_{tot} / \ell \geq \widetilde{\Omega} ( \log n )$, this would give an $\widetilde{\Omega} ( \log^2 n )$ bound on $t_{tot}$.

To show such a lower bound on $t_{tot} / \ell$, Larsen \cite{larsen_cell_2012} utilizes the cell-sampling method~\cite{siegel_universal_2004, panigrahy_lower_2010}. The main idea behind cell-sampling is the following. If there exists a too-good-to-be-true (i.e., short) $t_{tot}$ query algorithm, then naively sampling some $p := ( \poly ( t_u w ) )^{-1}$ fraction of the cells would give a nontrivial advantage (of roughly $p^{t_{tot} / \ell}$) in answering the queries. This is nontrivial in the regime of interest: with $\ell = \log_{\beta} n = \log n / \log \beta$ and $\beta = ( t_u w )^{2}$, both $1/\ell$ and $\log (1/p)$ contribute a factor of $\log ( t_u w )$, so that
\begin{equation*}
    p^{t_{tot} / \ell} = \exp \left( - \Theta \left( \frac{t_{tot} \left( \log ( t_u w ) \right)^2}{\log n} \right) \right) = n^{-o(1)}
\end{equation*}
exactly when $t_{tot} = o ( ( \log n / \log ( t_u w ) )^2 )$ --- the bound we aim to prove, and the source of its quadratic form. An equivalent way to frame this argument is to find a small subset $C_0 \subset C_i$ such that $C_0$ along with $C_{<i}$ and $C_{>i}$ will answer $n^{-o(1)}$ fraction of possible queries, without having to know the whole $C_i$. We can generate such $C_0$ by simply sampling each cell in $C_i$ at random with probability $p$. By averaging, there must exist a setting of such $C_0$ which can answer $n^{-o(1)}$ fraction of possible queries, say $\cQ^* \subset \cQ$.

An elegant way to formulate the underlying combinatorial problem is through the following one-way communication simulation, as the ultimate goal is to show such one-way communication cannot exist.
\begin{Protocol}
    \begin{itemize}
        \item Fix $\vec{S}$ known to both Alice and Bob. Bob is given all the updates $\{ X^{(i)} \}_{i=1}^\ell$, but not the desired query $Q \in \cQ$.
        \item Alice is given all the updates, except the updates in epoch $i$ (i.e., $X^{(i)}$), and the desired query $Q$.
        \item Bob sends a one-way message $M$ of length $\abs{M} \leq p \cdot ( w t_u ) \cdot n_i + o ( n_i )$ bits to Alice. Then Alice announces $f ( S_Q, X )$.
    \end{itemize}
    \caption{One-Way Communication Simulation \label{prot:one_way}}
\end{Protocol}

As a one-way message, Bob will send the small sub-sampled subset $\widetilde{C}_i$ along with $C_{<i}$. Cell-sampling ensures that $\widetilde{C}_i$ is small, while the Chronogram ensures that $C_{<i}$ is small. Alice then can simulate the queries for $\cQ^*$. So for any $Q \in \cQ^*$, Alice will be outputting the correct answer, achieving a nontrivial advantage! We then just need to show that a short one-way message $M$ cannot answer an $n^{-o(1)}$ fraction of $\cQ$.

Though tempting, this argument has a crucial flaw. $\widetilde{C}_i$ cannot identify $\cQ^*$ (i.e., the $\cQ^*$-identification problem). This occurs because Alice cannot distinguish whether the probed cells are inside $C_i \backslash \widetilde{C}_i$ or absent from $C_i$ entirely. As she cannot identify the set $\cQ^*$, there is no guarantee that Alice's guess provides zero advantage over completely random guessing for queries in $\cQ \backslash \cQ^*$. This technical subtlety arises in all subsequent works proving an $\omega( \log n )$ lower bound \cite{patrascu_lower_2007, larsen_cell_2012, larsen_crossing_2020, larsen_super-logarithmic_2025}.

Larsen \cite{larsen_cell_2012} tackled this technical subtlety by encoding a subset of $\cQ^*$ as a part of Bob's message, but this argument crucially requires the output to be $\Omega( \log n )$ bits. To extend the method to the Boolean setting, Larsen, Weinstein, and Yu~\cite{larsen_crossing_2020} introduced the so-called Peak-to-Average Lemma, a proof of which requires a technical dive into Chebyshev polynomials. The lemma allows Bob to carefully choose a subset of cells, knowledge of which yields a nontrivial advantage (of $n^{-o(1)}$) but with a $\sqrt{ \log n }$ parameter loss. No improvement can be made to the lemma, as the lemma is tight (Appendix B of \cite{larsen_crossing_2020}).

\subsubsection{Previous Work on the Multiphase Problem}

\Patrascu~envisioned a bold approach \cite{patrascu_mihai_towards_2010, thorup_mihai_2013} to break out of the Chronogram method and its extensions (all the methods in \pref{sec:polylog}), providing a new angle on dynamic Boolean cell-probe lower bounds.

His ambitious goal was to directly prove a polynomial lower bound by considering \pref{prob:Multiphase}. He conjectured that \pref{prob:Multiphase} requires $\max \{ t_u, t_{tot} \} \geq n^{\Omega(1)}$ for $m = n^{1+\Omega(1)}$, even when $f$ is Disjointness. This is the notorious Multiphase Conjecture. \Patrascu~then showed reductions to various dynamic problems such as dynamic reachability, among others.

It is worth noting that \Patrascu~\cite{patrascu_mihai_towards_2010} initiated the whole area of Fine-Grained Complexity of Dynamic Problems using the Multiphase Problem (\pref{prob:Multiphase}). He showed that the 3SUM Conjecture\footnote{The conjecture states that the 3SUM problem on $n$ integers from $\{ - n^4, \ldots, + n^4 \}$ cannot be solved in $O(n^{2-\eps})$ time for any $\eps > 0$.}~implies the Multiphase Conjecture. We refer the reader to the survey~\cite{williams_virginia_vassilevska_fine-grained_2017} for the context of the Multiphase Conjecture in the area of Fine-Grained Complexity of Dynamic Problems.

While the problem is conjectured to be polynomially hard, the best unconditional lower bound stood at merely $\Omega ( \log n )$ for over a decade against the weaker bit-probe model (i.e., the word size $w = 1$) \cite{brody_adapt_2015, clifford_new_2015, ko_adaptive_2020}.\footnote{While \cite{brody_adapt_2015, clifford_new_2015, ko_adaptive_2020} do not explicitly mention the lower bound, one can simulate a $t$-round adaptive algorithm with a $2^t$-round nonadaptive algorithm in the bit-probe model.} \Patrascu's proposed approach was to consider the following communication game, which is now commonly referred to as the Multiphase Communication Game.
\begin{Protocol}
    \begin{itemize}
        \item Alice holds $\vec{S} \in \left( \{0,1\}^{n} \right)^m$, and $Q \in [m]$. Bob holds $X \in \{0,1\}^n$ and $Q \in [m]$. Merlin holds $\vec{S}$ and $X$.
        \item Merlin sends $U$ to Bob of length $s = O(n \cdot t_u \cdot w)$.
        \item Alice and Bob engage in standard two-party communication of $w \cdot t_{tot}$ bits to output $f(S_Q, X )$.
    \end{itemize}
    \caption{ Multiphase Communication Game \label{prot:multiphase_comm} }
\end{Protocol}
The communication model can then simulate the update and query algorithms, as Merlin can simulate all the update algorithms, with $s = w \cdot n \cdot t_u$. Then Alice and Bob can simulate the query algorithm using $w \cdot t_{tot} = \widetilde{O}(t_{tot})$ bits of communication to output $f(S_Q, X )$.

In fact, as pointed out in \cite{ko_adaptive_2020, ko_lower_2025}, even a lower bound against 3-round communication between Alice and Bob (Alice, Bob, then Alice speak) would imply a breakthrough in circuit complexity, namely circuit lower bounds for random linear operators \cite{jukna_circuits_2010, jukna_boolean_2012, jukna_complexity_2013, drucker_limitations_2012}. Nevertheless, Ko and Weinstein~\cite{ko_adaptive_2020} developed an information-theoretic framework to handle 2-round communication (Alice speaks, then Bob) in the zero-error or small-error regime. This framework was further extended by Ko~\cite{ko_lower_2025} to handle the low-advantage regime, where the performance of Bob's final output is compared against a random guess.

Recently, Ko~\cite{ko_unifying_2025} combined the information-theoretic framework~\cite{ko_adaptive_2020, ko_lower_2025} with the one-way simulation theorem of Larsen, Weinstein, and Yu \cite{larsen_crossing_2020}, to give an $\widetilde{\Omega} ( \log^{3/2} (n))$ lower bound for \pref{prob:Multiphase} when $f$ is Inner Product over $\Field_2$. Since that work applied the one-way simulation theorem of Larsen, Weinstein, and Yu \cite{larsen_crossing_2020} directly and then used information-theoretic methods to lower bound the one-way communication, it could yield at most $\widetilde{\Omega} ( \log^{3/2} (n))$.

\subsection{Technical Contribution: The 2.5-Round Communication Game} \label{sec:technical}

We resolve the Boolean hardness question by introducing a fundamentally different
communication model that eliminates the technical barriers faced by all previous
approaches.

\subsubsection{The Core Challenge}
All previous super-logarithmic lower bounds \cite{larsen_cell_2012, larsen_crossing_2020, larsen_super-logarithmic_2025} rely on
the Chronogram and cell-sampling techniques, which reduce the problem to proving
that no short \emph{one-way} message from Bob to Alice can help answer an
$n^{-o(1)}$ fraction of queries in \pref{prot:one_way}. However, as discussed in \pref{sec:polylog}, there is a fatal technical subtlety: Alice does not know \emph{which queries} she can answer correctly (the set $\cQ^*$), because she cannot distinguish between:
\begin{itemize}
\item probed cells missing from Bob's randomly sampled subset;
\item probed cells that were never written to memory at all.
\end{itemize}

This ambiguity destroys the advantage argument. Larsen, Weinstein, and Yu
\cite{larsen_crossing_2020} resolved this using the Peak-to-Average Lemma, a sophisticated
application of Chebyshev polynomials, but this incurs a $\sqrt{ \log n }$ parameter loss,
yielding only $\widetilde{ \Omega} ( \log^{1.5} (n))$. The lemma is provably tight \cite[Appendix B]{larsen_crossing_2020}.

\subsubsection{Our Key Insight}

The one-way communication prevents Alice from verifying whether her simulation
is correct. But if we add just \emph{an additional round} --- letting Bob verify and certify
Alice's guess --- the entire technical barrier disappears.

\begin{figure}[!t]
\centering
\resizebox{0.8\textwidth}{!}{%
\begin{tikzpicture}[
  node distance=1.8cm,
  box/.style={rectangle, draw, thick, minimum width=2.8cm, minimum height=1.1cm, align=center},
  arrow/.style={->, very thick, >=stealth}
]

  \node[box, fill=purple!20] (merlin) {\textbf{Merlin} \\ $(\vec{S}, X)$};
  \node[box, fill=blue!20, below=of merlin] (bob) {\textbf{Bob} \\ $(X)$};
  \node[box, fill=red!20, right=4cm of bob] (alice) {\textbf{Alice} \\ $(\vec{S} , Q)$};

  \draw[arrow, purple!70!black] (merlin) -- node[left, align=right] {
    \textbf{Round 0} \\[2pt]
    $U$ (updates)
  } (bob);

  \draw[arrow, blue!70!black] (bob) -- node[above, align=center] {
    \textbf{Round 0.5} \\[2pt]
    $\Phi$: sampled cells
  } (alice);

  \draw[arrow, red!70!black] ([yshift=-1.1cm]alice.south) --
    node[below, align=center] {
      \textbf{Round 1} \\[2pt]
      $A_Q$: query transcript \\
      {\small (Alice's simulation)}
    } ([yshift=-1.1cm]bob.south);

  \node[box, fill=green!30, below=3.2cm of bob, minimum height=1.3cm] (output) {
    \textbf{Round 2} \\[2pt]
    Bob \textbf{verifies} $A_Q$ \\
    Outputs answer or \\
    random guess (FAIL)
  };
  \draw[arrow, green!60!black] ([yshift=-1.1cm]bob.south) -- (output);

\end{tikzpicture}%
}
\caption{
  Our 2.5-round Multiphase Communication Game. The crucial difference from prior work is the \textbf{verification step} (Round 2): Bob checks if Alice's transcript $A_Q$ is consistent with the actual memory $U$. The 0.5-round message $\Phi$ is sent independently of query $Q$, making this analyzable with information-theoretic techniques from \cite{ko_adaptive_2020, ko_lower_2025}.
}
\label{fig:2.5protocol}
\end{figure}
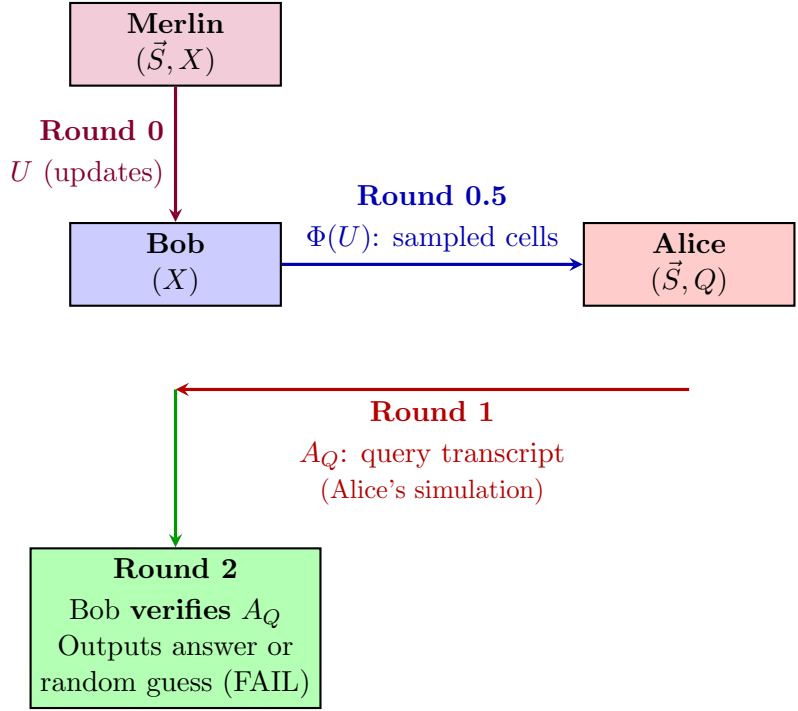

We introduce a \emph{2.5-round Multiphase Communication Game}:
\begin{enumerate}
\item \textbf{Round 0 (Merlin $\to$ Bob):} Merlin (holding $\vec{S},X$) sends update information $U$ to Bob.

\item \textbf{Round 0.5 (Bob $\to$ Alice):} Bob (holding $X$) forwards a message $\Phi$ to Alice (which may depend on $U$, $X$ and Bob's randomness, but not on the query).

\item \textbf{Round 1 (Alice $\to$ Bob):} The query $Q$ is revealed to Alice; Alice (holding $\vec{S}$ and $Q$) sends her simulation transcript $A_Q$ to Bob.

\item \textbf{Round 2 (Bob outputs):} Bob verifies $A_Q$ against $U$ and outputs the final answer.
\end{enumerate}
The protocol structure is illustrated in \pref{fig:2.5protocol}.

\paragraph{Verification Step}
The verification step is crucial:
\begin{itemize}
\item Alice simulates the query algorithm and sends the full transcript to Bob.
\item Bob checks if Alice's simulation is \emph{consistent} with the actual updated memory $U$.
\item If the query probes a cell of $C_i$ that Bob did not sample, Alice can only supply its content from before epoch $i$, reconstructed from $\vec{S}$ and $X^{(>i)}$; if that content is out of date, her transcript does not match $U$.
\item Bob detects this inconsistency and outputs FAIL, which then results in a random guess.
\item Only when Alice's simulation is correct does Bob output a meaningful answer.
\end{itemize}

This completely eliminates the $\cQ^*$ ambiguity: Bob's verification determines whether to use Alice's simulation or default to a random guess, so there is no need for Alice to identify $\cQ^*$. We delegate verification to Bob. We no longer need the Peak-to-Average Lemma, and can achieve the full $\Omega ( (\log n / \log \log n)^2 )$ bound.

\paragraph{Detectable Failure}
What the verification round buys, technically, is that the players \emph{know when the error occurs}. Bob either accepts, having checked Alice's entire transcript against the true memory, or declares FAIL; there is no third case in which the protocol is wrong without knowing it. In the one-way setting of \cite{larsen_crossing_2020} no such indicator is available: Alice's answers on $\cQ \backslash \cQ^*$ are uncontrolled, and a failure is indistinguishable from a success. What our proof uses is only the resulting decomposition. Conditioned on $\neg \text{FAIL}$ Bob is correct with certainty, and conditioned on FAIL he outputs an independent random guess and contributes nothing, so the advantage is exactly $\Pr [ \neg \text{FAIL} ]$. This is why advantage over a random guess, rather than error probability, is the quantity we track throughout (\pref{def:advantage}). That the players \emph{know} when the error occurs is not itself used anywhere in the argument; whether this additional structure can be exploited to strengthen the bound we leave as future work.

\paragraph{Our Second Contribution: Handling 2.5-Round Protocols}
Conceptually, the above simulation is simple. Yet, for the past decade, we have lacked the necessary technical tools to prove lower bounds for such 2.5-round protocols against small nontrivial advantage.

While 3-round communication lower bounds remain beyond current techniques (and
would imply major circuit lower bounds \cite{jukna_circuits_2010, drucker_limitations_2012, jukna_complexity_2013, dvir_static_2019}),
our key observation is that the information-theoretic framework introduced by Ko and Weinstein~\cite{ko_adaptive_2020}, further refined by Ko~\cite{ko_lower_2025} for the small-advantage regime, can handle exactly this 2.5-round structure. The key is that Bob's ``0.5-round'' message $\Phi$ is independent of the query $Q$. This allows us to apply the combinatorial lemmas from \cite{ko_adaptive_2020, ko_lower_2025}, and we can extract a random variable $Z$ from the 2.5-round protocol such that (a) the min-entropy of $S_Q$ and $X$ conditioned on $Z$ is large, and (b) $I( S_Q ; X | Z )$ is small.

\subsubsection{The Complete Argument}
Our proof proceeds in two steps:
\begin{enumerate}
\item \textbf{Simulation (\pref{sec:simulation}):} We show that any
      dynamic data structure with $w \cdot t_u = n^{o(1)}$ and $t_{tot} = o((\log n/\log ( w \cdot t_u ))^2)$ can be simulated by a 2.5-round Multiphase Communication Game with nontrivial advantage $n^{-o(1)}$.

\item \textbf{Communication Lower Bound (\pref{sec:comm_lb}):} We prove
      that no such 2.5-round protocol can achieve advantage better than
      $n^{-\Omega(1)}$, yielding a contradiction.
\end{enumerate}

We note that the two main components of our proof --- the simulation theorem (\pref{sec:simulation}) and the communication lower bound (\pref{sec:comm_lb}) --- are self-contained and can be verified independently.

\subsection{Why Is \texorpdfstring{$\widetilde{\Omega} ( \log^2 (n) )$}{Ω(log² n)} the Barrier?}

A natural question arises:
\begin{displayquote}
Can this approach be extended to prove higher lower bounds for dynamic Boolean problems, i.e., $\omega ( \log^2 n )$?
\end{displayquote}
We argue that our result likely matches the ``structural ceiling'' of the current techniques, specifically the \textit{Chronogram} framework.

The core of the Chronogram technique relies on decomposing the timeline of updates into a hierarchy of $\ell$ distinct epochs, where the duration of epoch $i$ grows geometrically. This construction naturally yields a hierarchy of depth $\ell = \Theta \left( \log n / \log ( t_u w ) \right)$, which is $\Theta \left( \log n / \log \log n \right)$ for $t_u = \poly \log (n)$. The lower bound is then obtained by constructing a hard distribution of updates and queries that forces the data structure to probe $\widetilde{\Omega}( \log n )$ distinct cells at \textit{each} epoch. We stress that the advance of this paper is not an exception to this. Our change is to the communication model that the simulation produces (\pref{sec:technical}), not to the epoch decomposition; the hierarchy, and hence the factor $\ell$, is inherited unchanged from \cite{fredman_cell_1989}.

Another perspective on the Chronogram framework is that it reduces dynamic cell-probe lower bounds to static cell-probe lower bounds (with pre-initialized memory and cache \cite{larsen_super-logarithmic_2025}), amplifying the static bound by a factor of $\ell$. As $\ell$ cannot exceed $\Theta ( \log n / \log \log n )$, to show an $\omega ( \log^2 n )$ lower bound we must show an $\omega ( \log n )$ lower bound for a single epoch $i \in [\ell]$.

Unfortunately, the barriers to static cell-probe lower bounds are well-known. The best known {\em explicit} static cell-probe lower bound is subject to the well-known logarithmic barrier \cite{miltersen_data_1998, siegel_universal_2004, patrascu_logarithmic_2006, panigrahy_geometric_2008, panigrahy_lower_2010, patrascu_unifying_2011, korten_stronger_2025}
\begin{equation} \label{eq:sampling}
t \geq \Omega \left( \frac{ \log \frac{|\cQ|}{n}}{ \log \frac{s}{n} } \right)
\end{equation}
where $s$ is the number of cells used by the static cell-probe data structure, and $|\cQ|$ is the number of possible queries. Even for {\em semi-explicit} static cell-probe lower bounds, where queries may be nonexplicit linear functions, only a lower bound of
\begin{equation} \label{eq:bit_probe_static}
t \geq \Omega ( \min \{ \log (|\cQ|/s) , n / \log s \} )
\end{equation}
is known, and only against the weaker bit-probe model \cite{ko_lower_2025}. Both \pref{eq:sampling} and \pref{eq:bit_probe_static} are merely logarithmic in our regime of interest, since both $s$ and $\abs{\cQ}$ are polynomial in $n$. And due to a well-known connection to circuit lower bounds \cite{jukna_circuits_2010, jukna_boolean_2012, jukna_complexity_2013, drucker_limitations_2012, dvir_static_2019, korten_stronger_2025, ko_lower_2025}, showing $t \geq \omega ( \log n )$ when $|\cQ|$ is $\poly(n)$ would be a major circuit complexity breakthrough.

To conclude, any $\omega ( \log^2 n )$ Boolean lower bound would require either completely breaking away from the \textit{Chronogram} framework, which has been the main recipe for more than 35 years, or a major circuit complexity breakthrough.

\begin{remark}[Natural Proofs]
    In recent work, Kouck\'{y}, Loff, Molli and Saks \cite{koucky_natural_2026} propose a \emph{natural proofs} barrier for data structure lower bounds, and argue --- conditional on the existence of suitable pseudorandom constructions --- that it takes effect at $\widetilde{\Omega} ( \log^2 n )$ in the dynamic setting. Whether the argument given here is natural in their sense is unclear to us, and we make no claim either way. We observe only that the threshold they identify is the same $\widetilde{\Omega} ( \log^2 n )$ that the \textit{Chronogram} framework has long been conjectured to impose, and that the discussion above reaches from a purely structural direction.
\end{remark}

\section{Preliminaries}

\subsection{Information Theory}

In this section, we provide the necessary background on information theory that is used
in this paper. For comprehensive background, we direct the reader to \cite{cover_elements_2006}. Unless noted otherwise, all logarithms are base 2.

\begin{definition}[Entropy]
	The entropy of a random variable $X$ is defined as
	\begin{equation*}
	H(X) := \sum_{x} \Pr[X=x] \log \frac{1}{\Pr[X=x]}.
	\end{equation*}
	Similarly, writing $p(x|y) := \Pr[X = x | Y = y]$, the conditional entropy is defined as
	\begin{equation*}
	H(X|Y) := \E_{Y} \left[ \sum_{x} p(x|y) \log \frac{1}{p(x|y)} \right].
	\end{equation*}
\end{definition}

\noindent We will also apply $H ( \cdot )$ to events: for an event $E$, we write $H(E)$ for the entropy of its indicator random variable $\mathbbm{1} [ E ]$, and likewise for conditional entropies. In particular, $H ( E | Y ) = 0$ asserts that $E$ is determined by $Y$.

\begin{fact}[Conditioning Decreases Entropy] \label{fact:conditioningentropy}
For any random variables $X$ and $Y$
\begin{equation*}
    H(X) \geq H(X|Y)
\end{equation*}
\end{fact}

\noindent With entropy defined, we can also quantify the correlation between two random variables, or how much information one random variable conveys about the other.

\begin{definition}[Mutual Information]
	Mutual information between $X$ and $Y$ (conditioned on $Z$) is defined as
	\begin{equation*}
	I( X; Y | Z) := H(X|Z) - H(X|YZ).
	\end{equation*}
\end{definition}

\noindent Similarly, we can quantify how far one distribution is from another.

\begin{definition}[KL-Divergence]
		The KL-divergence between two distributions $\mu$ and $\nu$ is defined as
	\begin{equation*}
	\KL{\mu}{\nu} := \sum_{x} \mu(x) \log \frac{\mu(x)}{\nu(x)}.
	\end{equation*}
\end{definition}

\noindent To bound mutual information, it suffices to bound KL-divergence, due to the following fact.

\begin{fact}[KL-Divergence and Mutual Information] \label{fact:KL_Mutual}
		The following equality between mutual information and KL-divergence holds:

	\begin{equation*}
	I(A;B|C) = \E_{B,C} \left[ \KL{ A|_{B=b, C=c} }{ A|_{C=c} } \right].
	\end{equation*}
\end{fact}

\begin{fact}[Pinsker's Inequality] \label{fact:pinsker}
For any two distributions $P$ and $Q$,
\begin{equation*}
\norm{ P - Q }_{TV} = \frac{1}{2} \norm{ P - Q }_1 \leq \sqrt{ \frac{1}{2 \log e} D( P || Q) }
\end{equation*}
\end{fact}

\noindent We also make use of the following facts on mutual information throughout the paper.

\begin{fact}[Chain Rule] \label{fact:chainrule}
For any random variables $A,B,C$ and $D$
\begin{equation*}
    I(AD;B|C) = I(D;B|C) + I(A;B|CD).
\end{equation*}
\end{fact}

\begin{fact} \label{fact:chainrule1}
For any random variables $A,B,C$ and $D$, if $I(B;D|C) = 0$
\begin{equation*}
    I(A;B|C) \leq I(A;B|CD).
\end{equation*}
\end{fact}
\begin{proof} By the chain rule and nonnegativity of mutual information,
\begin{align*}
I(A;B|C) \leq I(AD;B|C) & = I(B;D|C) + I(A;B|CD) \\
& = I(A;B|CD).
\end{align*}
\end{proof}

\begin{fact}\label{fact:chainrule2}
For any random variables $A,B,C$ and $D$, if $I(B;D|AC) = 0$
\begin{equation*}
    I(A;B|C) \geq I(A;B|CD).
\end{equation*}
\end{fact}
\begin{proof} By the chain rule and nonnegativity of mutual information,
\begin{align*}
I(A;B|CD) \leq I(AD;B|C) & = I(A;B|C) + I(B;D|AC) \\
& = I(A;B|C).
\end{align*}
\end{proof}

We will also use the following version of the Chernoff bound.
\begin{fact}[Chernoff Bound] \label{fact:chernoff}
 Let $X_1, \ldots , X_n$ be $n$ independent random variables with $X_i \in \{ 0 , 1 \}$ and $\Pr[ X_i = 1 ] = p$ for all $i \in [n]$. Then for any $\eps > 0$ with $\eps < 1 - p$, we have
 \begin{align*}
        \Pr \left[ \sum_{i \in [n]} X_i \geq ( p + \eps ) n \right] & \leq \exp \left(- \frac{n}{\log_2 e} \cdot D \left( \Ber ( p + \eps ) || \Ber ( p ) \right) \right)
    \end{align*}
 where $\Ber(p)$ refers to the Bernoulli distribution with success probability $p$.
\end{fact}

\subsection{Min-Entropy}

\begin{definition} We define the R\'enyi entropy $H_2 (A)$ and min-entropy $H_\infty (A)$ as
    \begin{align*}
        H_2 (A) := - \log \left( \sum_{a} \Pr[ A = a]^2 \right) \\
        H_\infty (A) := - \log \left( \max_{a} \Pr[ A = a]  \right)
    \end{align*}
\end{definition}

\begin{fact}[R\'enyi Entropy] \label{fact:renyi}
Let $A$ be a random variable. Then
    \begin{equation*}
        H ( A ) \geq H_2 (A) \geq H_\infty (A)
    \end{equation*}
 In particular, for any fixed $b$ we have
\begin{equation*}
H_2 ( A| B = b ) \geq H_\infty ( A | B = b)
\end{equation*}
\end{fact}

We use the following notion of ``average'' min-entropy.

\begin{definition}[Average Min-Entropy]
    \begin{align*}
        \widetilde{H}_\infty ( A | B ) & := - \log \left( \E_{b \sim B} \left[ \max_{a} \Pr [ A = a | B = b ] \right] \right) \\
        & \phantom{:}= - \log \left( \E_{b \sim B} \left[ 2^{- H_\infty (A | B = b ) } \right] \right)
    \end{align*}
\end{definition}

\begin{lemma}[Lemma 2.2 of \cite{dodis_fuzzy_2008}] \label{lem:dors}
        Let $A,B,C$ be random variables. If $B$ has at most $2^{\lambda}$ possible values, then
    \begin{equation*}
        \widetilde{H}_\infty ( A | B, C ) \geq \widetilde{H}_\infty ( A, B | C ) - \lambda \geq \widetilde{H}_\infty ( A | C ) - \lambda.
    \end{equation*}
    In particular,
    \begin{equation*}
        \widetilde{H}_\infty ( A | B ) \geq \widetilde{H}_\infty (A, B ) - \lambda \geq H_\infty ( A ) - \lambda.
    \end{equation*}
\end{lemma}

\begin{claim} \label{cl:conditioning_infty}
    \begin{equation*}
        \widetilde{H}_\infty ( A | B,C ) \leq \widetilde{H}_\infty ( A | B  )
    \end{equation*}
\end{claim}

\begin{claim}[\cite{jain_property_2009}] \label{cl:lowL2}
        Let $\cD$ be a distribution and $\cU$ the uniform distribution over some $\cS \subseteq \{0,1\}^n$. If $D( \cD || \cU ) \leq t$ with $t > 1$, then for every $\alpha \geq 4$ there exists an event $E$ such that
    \begin{align*}
        & \cD ( E ) \geq 1 - \frac{1}{\alpha} \\
        & H_\infty ( \cD |_{E} ) \geq  \log |\cS| - 2 \alpha t
    \end{align*}
\end{claim}

\noindent We include proofs of \pref{cl:conditioning_infty} and \pref{cl:lowL2} in \pref{app:missing} for completeness.

\section{Main Proof} \label{sec:main_proof}

In this section, we prove our main theorem.

\begin{reptheorem}{thm:main_informal}
    Let $f$ be Inner Product over $\Field_2$ and let $m = n^{1+\Omega(1)}$. Then any data structure for the Multiphase Problem for $f$ whose update time satisfies $t_u = n^{o(1)}$ must have
        $$  t_{tot} \geq \Omega \left( \left( \frac{\log n}{\log ( t_u w )} \right)^2 \right).$$
    In particular,
    $$ \max \{ t_u, t_{tot} \} \geq \Omega \left( \left( \frac{\log n}{\log \log n} \right)^2 \right). $$
\end{reptheorem}

Before we delve into the proof of the main theorem, we introduce formal notation and the model necessary for the proof.

\paragraph{Epochs} To give a super-logarithmic lower bound, the first main technical component is the Chronogram technique \cite{fredman_cell_1989}, where we divide the sequence of $n$ updates into epochs, epoch $i$ containing $n_i := \beta^i$ updates. We take $\ell$ to be the largest integer with $\sum_{i=1}^{\ell} n_i \leq n$, appending any remaining updates to the last epoch, so that all $n$ updates are covered. By maximality of $\ell$, $n < \sum_{i=1}^{\ell+1} \beta^i \leq \frac{\beta}{\beta - 1} \beta^{\ell+1} \leq 2 \beta^{\ell+1}$ for $\beta \geq 2$, which gives
\begin{equation} \label{eq:epoch_count}
    \beta^{\ell+1} > \frac{n}{2}, \qquad \ell = \Theta \left( \log_{\beta} n \right) = \Theta \left( \frac{\log n}{\log \beta} \right).
\end{equation}

For each epoch of $n_i$ updates, denoted as $X^{(i)}$, the update algorithms will use at most $t_u n_i$ cell-probe operations (write operations) to overwrite the memory. We denote the epochs of updates as $\{ U_i \}_{i=1}^\ell$, each of size $\abs{U_i} \leq 2 w \cdot t_u n_i$ (an address and a content of $w$ bits for each written cell), and we process the updates in reverse order, that is, the larger epochs get processed first.

The essence of the Chronogram technique is that every query algorithm can be decomposed in the following manner. Let $C_i$ denote the cells last updated in epoch $i$. The query algorithm for the query $q$ will then probe into $C_1, \ldots, C_{\ell}$. We denote by $t_q [i]$ the number of probes the query for $q \in \cQ$ makes into $C_i$. Then it must be the case that
\begin{equation*}
    \sum_{i=1}^{\ell} t_q [i] \leq t_{tot}
\end{equation*}

\paragraph{Communication Model} We define the following $2.5$-round Multiphase Communication Game, as an extension of the 1.5-round Multiphase Communication Game introduced in \cite{chattopadhyay_upper_2016}. This replaces the one-way communication model used in \cite{larsen_crossing_2020, larsen_super-logarithmic_2025} for the lower bound. One difference from \pref{prot:multiphase_comm} is worth noting: there the query is known to both players, whereas here only Alice receives it. Our model sits between two extremes. \Patrascu's model \cite{patrascu_mihai_towards_2010} has so far resisted every attempt at a lower bound, whereas the restricted models in which lower bounds have been obtained \cite{chattopadhyay_upper_2016, ko_adaptive_2020} do not imply lower bounds for general dynamic data structures. The 2.5-round model is weaker than \Patrascu's, yet strong enough to yield such an implication, and it is what the simulation of \pref{sec:simulation} actually produces.

\begin{Protocol}
    \begin{itemize}
        \item Alice holds $\vec{S}$. Bob holds $X$. Merlin holds $\vec{S}$ and $X$.
        \item Merlin sends $U$ to Bob.
        \item As the 0.5-round message, Bob simply forwards to Alice a sub-sampled partial view $\Phi$ of the updated memory. It may depend on $U$, on Bob's own input $X$ and on his private randomness, but never on the query.
        \item The query $q \in \cQ$ is then revealed to Alice. Alice sends a message $A_q$ to Bob, and Bob outputs an estimate $B_q \in \{ \pm 1 \}$ of $f( S_q, X )$ as a function of $U$, $A_q$ and his own randomness.
    \end{itemize}
    \caption{ 2.5-Round Multiphase Communication Game \label{prot:two_five} }
\end{Protocol}

\paragraph{Advantage of the Protocol}

We will be considering a fixed $\vec{S}$ for the simulation theorem, while for the lower bound analysis we will assume that these $\vec{S}$ are distributed uniformly at random. Regardless of the chosen $\vec{S}$, we will be using a uniform distribution over $X$ as the update, and a uniform distribution over $\cQ$ as the possible query. We do not require $f$ to be balanced: the only property of $f$ used in the lower bound is that it has small bias under product distributions of large min-entropy (\pref{sec:comm_lb}; \pref{def:hard} for general $f$). This already accounts for rare biased inputs, such as $S_Q = 0$ for Inner Product over $\Field_2$. Under these distributions, we measure the advantage of the 2.5-round protocol $\Pi$. Throughout the rest of the paper we abbreviate the messages exchanged on query $q$ as
$$\cM_q := ( \Phi, A_q, B_q ),$$
that is the whole $2.5$-round transcript conversed between Alice and Bob.
\begin{definition} \label{def:advantage}
    Fix $\vec{S}$. The advantage of a 2.5-round protocol $\Pi = ( \Phi, A_{Q}, B_{Q} )$ under the distribution $Q$ over the queries is defined as
    \begin{equation*}
        \adv_{\vec{S}} ( G_{f}, \Pi ) := \E_{Q, \cM_Q} \left[ \E_{(S_Q,X) | Q, \cM_Q} \left[ B_Q \cdot f ( S_Q , X ) \right] \right]
    \end{equation*}
        where $B_Q \in \{ \pm 1 \}$ is Bob's final estimate.
\end{definition}

For the rest of the section, we prove the simulation theorem (\pref{sec:simulation}) showing that any efficient data structure yields an efficient 2.5-round protocol achieving good advantage. This results in a contradiction via information-theoretic arguments (\pref{sec:comm_lb}).

\subsection{Simulation Theorem} \label{sec:simulation}

First, we show the following simulation theorem, which argues that an efficient dynamic data structure implies an efficient 2.5-round Multiphase Communication Game protocol $\Pi$ with large advantage. Throughout this section we write
\begin{equation} \label{eq:xi}
        \xi ~:=~ \frac{100 \, t_{tot}}{\ell} \, \ln (1/p)
\end{equation}
for the cost, in the exponent, of the event on which the simulation succeeds: Bob's sample $\widetilde{C}_i$ containing every cell of $C_i$ that the query probes.

\begin{theorem}[New Simulation Theorem] \label{thm:simulation}
        Fix $\vec{S}$, put $\beta := ( t_u w )^{2}$, and let $\ell = \Theta ( \log n / \log \beta )$ be the resulting number of epochs as in \pref{eq:epoch_count}. If there exists a data structure for the Multiphase Problem for $f$ with query time $t_{tot}$ and update time $t_u$, then for every $p \in (0, 1/3)$ with $p \leq 1/(8 t_u w)$ there exists a 2.5-round Multiphase Communication Game protocol $\Pi$ with
    \begin{align*}
        \abs{U}        & \leq 2 w t_u n + \log \ell, \\
                \abs{\Phi} & \leq n - \frac{1}{2} \left( 1 - 4 p \, w t_u - o(1) \right) n^{2/3}, \\
        \abs{A_q}      & \leq O ( w \, t_{tot} ),
    \end{align*}
    achieving advantage
    \begin{equation*}
        \adv_{\vec{S}} ( G_{f}, \Pi ) ~\geq~ \Omega \left( e^{-\xi} - e^{-\Omega ( p \, n^{2/3} )} \right).
    \end{equation*}
\end{theorem}

\begin{proof}

        Suppose there exists a dynamic data structure for the Multiphase Problem for $f$ with $t_{tot}$ query time and $t_u$ update time. We simulate the data structure with a 2.5-round Multiphase Communication Game protocol as follows.

    \begin{figure}[!t]
    \centering
    \fbox{\begin{minipage}{0.97\textwidth}

    \begin{itemize}
        \item First Merlin chooses $i \in [2 ( \ell + 1 ) / 3, \ell]$, which satisfies
        \begin{equation*}
            \E_{ q \in \cQ | \vec{S}, X} \left[ t_q [i] \right] \leq 10 \cdot t_{tot} / \ell.
        \end{equation*}
        Note that by a simple averaging argument, such $i$ is guaranteed to exist, as the interval contains at least $( \ell - 2 ) / 3 \geq \ell / 10$ integers for $\ell \geq 3$. Merlin then forwards $U = U_1, \ldots, U_{\ell}$ along with the index $i$ to Bob. From $U$, Bob can construct $C_1, \ldots, C_\ell$.

        \item As the $0.5$-round message, Bob sends the following. For $C_i$, he samples each cell with probability $p$ at random, and creates $\widetilde{C}_i$. If $\widetilde{C}_i$ contains more than $2 p \cdot t_u n_i$ cells, then Bob simply aborts, and outputs FAIL. Otherwise, he forwards $X^{(>i)}, C_{<i}$ and $\widetilde{C}_i$ along with $i$ (which is evident in the message already) to Alice. Observe that the message depends only on $U$, $X$, $i$ and Bob's randomness, and not on $q$.

        \item Alice does the following.
        \begin{itemize}
        \item Using $\vec{S}$, and $X^{(>i)}$, the first $\ell - i$ epochs of updates, Alice constructs $U_{>i}$ from $X^{(>i)}$.
        \item Alice answers a probe to an address $a$ as follows. If $a \in \widetilde{C}_i \cup C_{<i}$, she uses the content supplied in $\Phi$. Otherwise she uses the content given by her own reconstruction from $\vec{S}$ and $X^{(>i)}$, which is correct whenever $a \notin C_i$. Note that in the latter case she cannot tell whether $a$ in fact lies in $C_i \backslash \widetilde{C}_i$, in which case her content is wrong; this is precisely the ambiguity that Bob's verification resolves.

        \item Alice simulates the query algorithm for $q \in \cQ$ under $\vec{S}$, generating the query transcript
        \begin{equation*}
            A_q := ( ( a_1, c_1 ), \ldots , ( a_{t_{tot}} , c_{t_{tot}} ), z_{out})
        \end{equation*}
        where $a_\tau$ refers to the address probed at time $\tau$, $c_\tau$ refers to its content, and $z_{out}$ refers to the final output of the query algorithm.
        \item Alice sends the query transcript $A_q$ to Bob.
        \end{itemize}
                \item Then Bob simply verifies $A_q$. That is, he checks whether $(a_\tau, c_\tau)$ matches $U$ for every $\tau \in [t_{tot}]$. (Note that if the cell $a_{\tau}$ is never updated by $U$, Alice's content $c_\tau$ must be correct, as Alice holds $\vec{S}$.) If every probe matches, he declares $z_{out}$ as the answer $B_q$. Otherwise, he outputs FAIL.
        \item Whenever Bob aborts or the verification fails, he outputs FAIL and an independent random guess as $B_q$.
    \end{itemize}

    \end{minipage}}

    \refstepcounter{Protocol}\label{prot:simulation}
    \vspace{4pt}
    {\small Protocol \theProtocol: Simulation Protocol $\Pi$}
    \end{figure}

    The cost guarantee of the protocol is straightforward. Merlin's message is simply the update together with an index $i \in [2 ( \ell + 1 ) / 3, \ell]$. Therefore $|U| \leq 2 t_u w n + \log \ell = O ( t_u w n )$. Each cell in $\Phi$ is sent as an address and a content, $2w$ bits. Next, we bound the length of Bob's 0.5-round message:
    \begin{align}
        \abs{\Phi}
        & = \abs{ X^{(>i)} } + \abs{C_{<i}} + \abs{\widetilde{C}_i} \nonumber \\
        & \leq \sum_{k > i} n_k + \sum_{k < i} ( 2 w t_u ) \, n_k + 4 p \, w t_u \, n_i \nonumber \\
        & \leq ( n - n_i ) + 4 p \, w t_u \, n_i + o ( n_i ) \label{eq:0.5_bound}
    \end{align}
        where the upper bound follows from the guarantee on the epoch sizes, $n_i := \beta^i = ( t_u w )^{2i}$.
    Here $\sum_{k < i} (2 w t_u) \cdot n_k \leq 4 w t_u \, n_{i-1} \leq 4 w t_u \, n_i / \beta = o(n_i)$ as it is a geometric series and $\beta = ( t_u w )^2$. As Merlin chooses $i \in [2 ( \ell + 1 ) / 3, \ell]$ and $\beta^{\ell+1} > n/2$ by \pref{eq:epoch_count}, we get $n_i \geq \beta^i \geq \beta^{2 ( \ell + 1 ) / 3} > ( n / 2 )^{2/3} \geq n^{2/3} / 2$ (with $\geq$ rather than $=$, as the leftover updates enlarge the last epoch). Therefore, assuming $1 - 4 p \cdot w t_u - o(1) > 0$, as we will choose the parameters accordingly, \pref{eq:0.5_bound} becomes
    \begin{align} \label{eq:0.5_final_bound}
        \pref{eq:0.5_bound}
        & \leq n - ( 1 - 4 p \, w t_u - o(1) ) \, n_i \nonumber \\
                & \leq n - \frac{1}{2} ( 1 - 4 p \, w t_u - o(1) ) \, n^{2/3}.
    \end{align}
    Next, we bound the length of Alice's message,
    \begin{equation} \label{eq:Alice_final_bound}
                \abs{A_q} = 2 w \cdot t_{tot} + 1 \leq 3 w \cdot t_{tot}.
    \end{equation}
    \pref{eq:0.5_final_bound} and \pref{eq:Alice_final_bound} give the desired bound on the cost of $\Pi$. We now lower bound $\adv_{\vec{S}} ( G_{f}, \Pi )$ by bounding the probability that Alice's simulation is consistent with $U$, so that Bob's verification accepts.
    \begin{claim} \label{cl:advantage_lower_bound}
        \begin{equation*}
        \adv_{\vec{S}} ( G_{f}, \Pi ) ~\geq~ \Omega \left( e^{-\xi} - e^{-\Omega ( p \, n^{2/3} )} \right)
        \end{equation*}
    \end{claim}
    \begin{proof}
        Recall that $\vec{S}$ is fixed and that Merlin chose $i \in [2 ( \ell + 1 ) / 3, \ell]$ with
        \begin{equation*}
            \E_{\cQ | \vec{S}, X} \left[ t_q [i] \right] \leq 10 \cdot t_{tot} / \ell.
        \end{equation*}
        We first note the observation that reduces the analysis to a probability estimate. If Bob's verification accepts, then $B_q = f ( S_q, X )$; if it rejects, $B_q$ is an independent random guess. The advantage on a query $q$ is therefore exactly $\Pr [ \text{verification accepts} ]$, and in particular it is never negative. It suffices, then, to exhibit a constant fraction of queries on which acceptance is likely.

        \emph{Step 1: discarding the heavy queries.} By Markov's inequality, at most a $1/10$ fraction of the queries have $t_q [i] > 100 \, t_{tot} / \ell$. We claim nothing on those; by the observation above they contribute a nonnegative amount. Fix a query $q$ with $t_q [i] \leq 100 \, t_{tot} / \ell$.

        \emph{Step 2: when does the verification accept?} Consider the execution of the query algorithm for $q$ on the actual memory, and let $P_q \subseteq C_i$ be the set of $t_q [i]$ cells of $C_i$ that it probes. Note that $P_q$ is determined by $\vec{S}, X$ and $i$ alone; in particular it does not depend on Bob's sampling. We claim that the verification accepts whenever $P_q \subseteq \widetilde{C}_i$. Indeed, every content Alice then supplies agrees with $U$: an address in $\widetilde{C}_i \cup C_{<i}$ is answered from $\Phi$, and any other address lies outside $C_i$, where her reconstruction from $\vec{S}$ and $X^{(>i)}$ is correct. Her simulation therefore reproduces the actual execution probe by probe, so Bob accepts.

        \emph{Step 3: the sampling.} Bob samples each cell of $C_i$ independently with probability $p$, and aborts only when $\widetilde{C}_i$ contains more than $2 p \cdot t_u n_i$ cells. Since $\abs{P_q} \leq 100 \, t_{tot} / \ell$,
        \begin{align*}
            \Pr [ \text{accepts} ]
            & \geq \Pr [ P_q \subseteq \widetilde{C}_i ] - \Pr [ \text{aborts} ] \\
            & \geq p^{100 t_{tot} / \ell} - e^{-\Omega \left( n_i \, D ( \Ber ( 2p ) || \Ber ( p ) ) \right)} \\
            & \geq e^{-\xi} - e^{-\Omega ( p \, n^{2/3} )},
        \end{align*}
        where the last step uses \pref{eq:xi}, the standard estimate $D ( \Ber ( 2 p ) || \Ber ( p ) ) = \Theta ( p )$, and $n_i \geq n^{2/3} / 2$. The abort bound is a Chernoff bound over the at most $t_u n_i$ cells of $C_i$.

        Averaging over the at least $9/10$ fraction of queries covered by Steps 2 and 3 gives
        \begin{equation} \label{eq:advantage_simplified}
            \adv_{\vec{S}} ( G_{f}, \Pi ) ~\geq~ \Omega \left( e^{-\xi} - e^{-\Omega ( p \, n^{2/3} )} \right)
        \end{equation}
        which then completes the proof of the claim.
    \end{proof}

    \pref{eq:0.5_final_bound}, \pref{eq:Alice_final_bound} and \pref{cl:advantage_lower_bound} complete the proof of the simulation theorem.

\end{proof}

\begin{remark}[Detected Error] \label{rem:one_sided}
    \pref{prot:simulation} errs only on an event the players detect: whenever Bob does not declare FAIL, he has verified Alice's transcript against $U$ in full, and therefore $B_q = f ( S_q, X )$ with certainty. Every loss in $\adv_{\vec{S}} ( G_{f}, \Pi )$ thus comes from $\Pr [ \text{FAIL} ]$, on which Bob outputs an independent random guess and contributes zero advantage. Consequently $\adv_{\vec{S}} ( G_{f}, \Pi ) = \Pr [ \neg \text{FAIL} ]$, and the lower bound of \pref{sec:comm_lb} need only bound this single quantity, in contrast to the one-way analyses of \cite{larsen_cell_2012, larsen_crossing_2020}.
\end{remark}

Now we would like to select a range of parameters to derive a contradiction while simplifying the expressions.

\begin{corollary}
    \label{cor:choice_parameters}
    Assume $t_u = n^{o(1)}$ and set $t_{tot} = o \left( \left( \frac{\log n}{ \log ( t_u w )} \right)^2 \right)$, $\beta = (w \cdot t_u )^{2}$, $ p := \beta^{-1} / 2 $. Then
    \begin{equation*}
                \abs{U} \leq n^{1 + o(1)} ,~\abs{\Phi} \leq n - \frac{n^{2/3}}{4} ,~\abs{A_q} \leq o ( \log^3 (n) )
    \end{equation*}
    while having
    \begin{equation*}
        \adv_{\vec{S}} ( G_{f}, \Pi ) \geq n^{-o(1)}.
    \end{equation*}
\end{corollary}
\begin{proof}
        Let $\ell$ be as in \pref{eq:epoch_count} with $\beta = (t_u w)^{2}$, and let the sampling probability be $p := \beta^{-1} / 2 = 1 / ( 2 ( t_u w )^2 )$, so that $p \leq 1 / ( 8 t_u w )$ and $4 p \, w t_u = 2 / ( t_u w ) = o(1)$. As $w = \Theta ( \log n )$, the hypothesis $t_u = n^{o(1)}$ gives $t_u w = n^{o(1)}$ and hence $\log \beta = o ( \log n )$, so that $\ell = \omega (1)$ and $\ell = \Theta ( \log n / \log ( t_u w ) )$. The cost bounds are then
    \begin{align*}
        & \abs{U} \leq O ( w n t_u ) \leq n^{1 + o(1)} \\
                & \abs{\Phi} \leq n - \frac{1}{2} ( 1 - o(1) ) \, n^{2/3} \leq n - \frac{n^{2/3}}{4} \\
        & \abs{A_q} \leq O ( t_{tot} \cdot w ) \leq o ( \log^3 n ),
    \end{align*}
    where the last bound uses $\log ( t_u w ) \geq \log w = \Omega ( \log \log n )$, whence $t_{tot} = o ( ( \log n / \log \log n )^2 )$, together with $w = \Theta ( \log n )$. For the advantage, $\ln ( 1/p ) = \ln ( 2 \beta ) = \Theta ( \log ( t_u w ) )$, so that
    \begin{equation*}
        \xi = \Theta \left( \frac{t_{tot} \left( \log ( t_u w ) \right)^2}{\log n} \right) = o ( \log n )
    \end{equation*}
    while $p \, n^{2/3} = n^{2/3 - o(1)}$, so that \pref{eq:advantage_simplified} is at least $\Omega ( n^{-o(1)} ) - e^{- n^{2/3 - o(1)}} \geq n^{-o(1)}$.
\end{proof}

\subsection{2.5-Round Multiphase Communication Game Lower Bound} \label{sec:comm_lb}

The goal of this section is to show that such an advantage is impossible, building on the framework of \cite{ko_adaptive_2020, ko_lower_2025, ko_unifying_2025}. The key idea, developed across \cite{ko_adaptive_2020, ko_lower_2025, ko_unifying_2025}, is to extract a random variable $Z$ containing the output $Z_{out}$ from the 2.5-round protocol such that (a) the min-entropy of $S_Q$ and $X$ conditioned on $Z$ is large, and (b) $I( S_Q ; X | Z )$ is small. Then it must be the case that $\E_{(S_Q,X)|Z} \left[ Z_{out} \cdot f(S_Q,X) \right]$ is small. Throughout this section we abbreviate the pair of protocol messages --- Bob's $0.5$-round message and Alice's messages \textbf{for all queries} --- as the \emph{transcript}
\begin{equation*}
    \cT ~:=~ \left( \Phi, A_{\cQ_{all}} \right),
\end{equation*}
and we follow the analogous argument, selecting $Z$ as
\begin{equation*}
    Z := \left( \cT, S_{< Q}, B_{Q}, Q \right)
\end{equation*}
where $Q$ is chosen uniformly at random from $\cQ = [m]$ with $m = n^{1 + \Omega(1)}$. $A_{\cQ_{all}}$ refers to $A_q$ for all $q \in \cQ$, and $S_{<Q}$ refers to $S_1, \ldots, S_{Q-1}$. Note that $B_{Q}$ plays the role of the output $Z_{out}$. We need to include $A_q$ for all $q \in \cQ$ to apply direct-sum techniques.

We first show that \pref{cor:choice_parameters} implies the following claim on the advantage conditioned on $Z$.
\begin{claim} \label{cl:advantage}
    \begin{equation*}
    \E_{Z} \left[ \abs{ \E_{(S_Q,X)|Z} \left[ \chi_{S_Q} (X) \right] } \right] \geq n^{-o(1)}
    \end{equation*}
    where $\chi_{S_Q} (X) := (-1)^{\ip{S_Q,X}}$.
\end{claim}
\begin{proof}
    As $Q$ is chosen uniformly at random, \pref{cor:choice_parameters} guarantees that for any fixed $\vec{S}$,
    \begin{equation*}
        \E_{\cM_Q, Q | \vec{S}} \left[ B_Q \cdot \chi_{S_Q} (X) \right] \geq n^{-o(1)}
    \end{equation*}
    This then implies
    \begin{align*}
        n^{-o(1)} & \leq \E_{\cM_Q, Q, \vec{S}} \left[ \E_{X | \cM_Q, Q, \vec{S} } \left[ B_Q \cdot \chi_{S_Q} (X) \right] \right] \\
        & = \E_{\cM_Q, Q, S_{<Q}, S_Q} \left[ \E_{ X | \cM_Q, Q, S_{<Q}, S_Q } \left[ B_Q \cdot \chi_{S_Q} (X) \right] \right] \\
        & = \E_{\cM_Q, Q, S_{<Q} } \left[ \E_{ (S_Q, X) | \cM_Q, Q, S_{<Q} } \left[ B_Q \cdot \chi_{S_Q} (X) \right] \right] \\
        & \leq \E_{\cM_Q, Q, S_{<Q}} \left[ \abs{ \E_{ (S_Q,X) | \cM_Q, Q, S_{<Q} } \left[ \chi_{S_Q} (X) \right] } \right]
    \end{align*}
    But then
    \begin{align*}
        n^{-o(1)} &\leq \E_{\cM_Q, Q, S_{<Q}} \left[ \abs{ \E_{ (S_Q,X) | \cM_Q, Q, S_{<Q} } \left[ \chi_{S_Q} (X) \right] } \right] \\
        & \leq \E_{\cT, B_Q, Q, S_{<Q}} \left[ \abs{ \E_{ (S_Q,X) |\cT, B_Q, Q, S_{<Q}  } \left[ \chi_{S_Q} (X) \right] } \right] \\
        & = \E_{Z} \left[ \abs{ \E_{(S_Q,X)|Z} \left[ \chi_{S_Q} (X) \right] } \right]
    \end{align*}
    where the inequality follows from Jensen's inequality on the convex function, $\abs{~\cdot~}$, and since $( \cT, B_Q, Q )$ determines $\cM_Q$.

\end{proof}

We now show that \pref{cl:advantage} leads to a contradiction.

\subsubsection{Small Information on \texorpdfstring{$S_Q$}{S\_Q}}

First, we would like to show that the min-entropy of $S_Q$ conditioned on $Z$ is large. The statement is not true in general. However, a simple Markov argument allows us to ``extract'' a good event $E_Q$. We start with the following standard claim on the mutual information.

\begin{claim} \label{cl:mutual_info_s_I}
    \begin{equation*}
        I ( S_Q ; \cT, S_{< Q}, Q ) \leq 4 w \cdot t_{tot}
    \end{equation*}
\end{claim}
\begin{proof}
        By a standard direct-sum argument, since $Q$ is chosen independently at random and $S_1, \ldots, S_m$ are i.i.d., we have
    \begin{align*}
         I ( S_Q ; \cT, S_{< Q}, Q ) & = \underbrace{I ( S_Q ;  S_{< Q}, Q )}_{=0} + I ( S_Q ; \cT | S_{< Q}, Q ) \\
        & = \E_{Q} \left[ I ( S_Q ; \cT | S_{< Q}) \right ]
    \end{align*}
    Then we can bound $\E_{Q} \left[ I ( S_Q ; \cT | S_{< Q}) \right ]$ as
    \begin{align*}
         \E_{Q} \left[ I ( S_Q ; \cT | S_{< Q}) \right ] & = \frac{1}{m} \cdot I ( \vec{S} ; \cT ) \\
        & \leq \underbrace{\frac{\abs{\Phi}}{m}}_{\leq \frac{n}{m} = n^{-\Omega(1)} = o(1)} + \frac{m \abs{A_{q}} }{m} \\
                & \leq o(1) + 3 w \cdot t_{tot} \leq 4 w \cdot t_{tot}
    \end{align*}
    where the first inequality follows from $I ( \vec{S} ; \cT ) \leq H ( \cT ) \leq \abs{\Phi} + m \abs{A_{q}}$, completing the proof of the claim.
\end{proof}

Unfortunately, bounds on the mutual information (equivalently, on the KL-divergence) alone do not suffice for our proof. Since we are considering the ``small'' advantage regime, we need to argue that the min-entropy of $S_Q$ conditioned on $Z$ is large, as in \cite{ko_lower_2025}. We will use the following lemma to extract a ``good'' event in terms of min-entropy.

\begin{lemma}[$S_Q$ Min-Entropy]
\label{lem:min_entropy_s_i}
    For every setting of the parameter $\alpha \in (0, \frac{1}{10} )$, there exist events $E_Q^1$ and $E_Q^2$ such that
    \begin{align*}
        H ( E_Q^1 | \cT, S_{< Q}, Q ) = H ( E_Q^2 | \cT, S_{< Q}, S_Q, Q ) = 0
    \end{align*}
    and
    \begin{align*}
        & \Pr \left[ E_Q^1 \right] \geq 1 - \alpha,~~\Pr \left[ E_Q^2 \mid \cT, S_{< Q}, Q, E_Q^1 {=} 1 \right] \geq 1 - \alpha \\
        & H_\infty ( S_Q | \cT, S_{< Q}, Q, E_Q {=} 1 ) \geq n - \frac{8 w \cdot t_{tot}}{\alpha^2}
    \end{align*}
        where we denote $E_Q^1 {=} 1, E_Q^2 {=} 1$ as $E_Q {=} 1$ for short. The min-entropy bound holds for every fixing of $( \cT, S_{< Q}, Q )$ consistent with $E_Q {=} 1$; in particular it also holds for the average min-entropy $\widetilde{H}_\infty ( S_Q | \cT, S_{< Q}, Q, E_Q {=} 1 )$, which is how we use it in \pref{lem:comm_lb}.
\end{lemma}
\begin{proof}
        First, we define the event $E_Q^1$ as the set of $(\cT, S_{< Q}, Q)$ such that
    \begin{equation*}
        D ( S_Q |_{\cT, S_{< Q}, Q} || S_Q ) \leq \frac{4 w \cdot t_{tot}}{ \alpha }
    \end{equation*}
    Due to \pref{cl:mutual_info_s_I},
    \begin{align*}
        I ( S_Q ; \cT, S_{< Q}, Q ) & = \E_{\cT, S_{< Q}, Q} \left[ D ( S_Q |_{\cT, S_{< Q}, Q} || S_Q ) \right] \\
        & \leq 4 w \cdot t_{tot}
    \end{align*}
        Markov's inequality implies $\Pr[ E_Q^1 ] \geq 1 - \alpha$. Conditioned on $E_Q^1 {=} 1$,
    we obtain the event $E_Q^2$ from \pref{cl:lowL2}. As $E_Q^1 {=} 1$ guarantees
    \begin{align*}
        D ( S_Q |_{\cT, S_{< Q}, Q, E_Q^1 {=} 1} || S_Q ) & = D ( S_Q |_{\cT, S_{< Q}, Q} || S_Q ) \leq \frac{4 w \cdot t_{tot}}{ \alpha }
    \end{align*}
    where the equality holds as $E_Q^1$ is a deterministic function of $\cT, S_{< Q}, Q$,
        \pref{cl:lowL2}, applied with $t := 4 w \cdot t_{tot} / \alpha$ and with $1 / \alpha > 10$ in place of $\alpha$, yields an event $E_Q^2$ such that
    \begin{align*}
        & \Pr[ E_Q^2 \mid \cT, S_{< Q}, Q, E_Q^1 {=} 1 ] \geq 1 - \alpha \\
        & H ( E_Q^2 | S_Q, \cT, S_{< Q}, Q, E_Q^1 {=} 1 ) = 0 \\
        & H_\infty ( S_Q | \cT, S_{< Q}, Q, E_Q^1 {=} 1, E_Q^2 {=} 1 ) \geq n - \frac{8 w \cdot t_{tot}}{\alpha^2}
    \end{align*}
    which completes the proof of the lemma.
\end{proof}

\subsubsection{Small Information on \texorpdfstring{$X$}{X}}

Next, we show that $\widetilde{H}_\infty ( X | Z, E_Q {=} 1)$ is large. We remark that the conditioning on $E_Q {=} 1$ is needed to derive the contradiction.

\begin{lemma} \label{lem:min_entropy_x}
    \begin{equation*}
        \widetilde{H}_\infty ( X | Z, E_Q {=} 1) \geq n - |\Phi| - 10 \alpha - 1
    \end{equation*}
\end{lemma}
\begin{proof}
    Recall that as $X$ is chosen uniformly at random, regardless of the choice of $Q$ and $\vec{S}$,
    \begin{align*}
        H_\infty ( X | \vec{S}, Q ) = n
    \end{align*}
    while \pref{lem:dors} implies that
    \begin{align*}
        \widetilde{H}_\infty ( X | \vec{S}, \Phi, Q ) \geq n - |\Phi| .
    \end{align*}
    But $\vec{S}, \Phi, Q$ and Alice's independent private randomness determine $A_{\cQ_{all}}$. That is, $A_{\cQ_{all}}$ is determined by Alice's private randomness and  $\vec{S}, \Phi$. Therefore,
    \begin{equation*}
        I ( X ; A_{\cQ_{all}} | \vec{S}, \Phi, Q ) = 0.
    \end{equation*}
        Moreover, $(\vec{S}, \cT, Q)$ contains $(S_Q, \cT, S_{< Q}, Q)$, which determines both $E_Q^1$ and $E_Q^2$. Therefore,
    \begin{align*}
        \widetilde{H}_\infty ( X | \vec{S}, \cT, E_Q, Q ) \geq n - |\Phi|
    \end{align*}
    Then by the definition of average min-entropy we have the following equality.
    \begin{align*}
        2^{- \widetilde{H}_\infty ( X | \vec{S}, \cT, E_Q, Q) } & = \Pr [ E_Q {=} 1 ] \cdot 2^{- \widetilde{H}_\infty ( X | \vec{S}, \cT, E_Q {=} 1, Q) } \\
        & + \Pr [ E_Q {=} 0] \cdot 2^{- \widetilde{H}_\infty ( X | \vec{S}, \cT, E_Q {=} 0, Q) }
    \end{align*}
    which implies the following inequality
    \begin{align}
        \widetilde{H}_\infty ( X | \vec{S}, \cT, E_Q {=} 1, Q ) & \geq \log \left( \Pr [ E_Q {=} 1] \right) + \widetilde{H}_\infty ( X | \vec{S}, \cT, E_Q, Q ) \label{eq:x_min_infinity}
    \end{align}
    The left-hand side of \pref{eq:x_min_infinity} is upper bounded by
    \begin{equation*}
        \widetilde{H}_\infty ( X | \vec{S}, \cT, E_Q {=} 1, Q ) \leq \widetilde{H}_\infty ( X | S_{<Q}, \cT, E_Q {=} 1, Q )
    \end{equation*}
    due to \pref{cl:conditioning_infty}.

    \pref{lem:min_entropy_s_i} guarantees that $\Pr[ E_Q {=} 1 ] \geq ( 1 -  \alpha )^2 \geq 1 - 2\alpha$. Then the right-hand side of \pref{eq:x_min_infinity} is lower bounded by
    \begin{align*}
        \log \left( \Pr [ E_Q {=} 1] \right) + \widetilde{H}_\infty ( X | \vec{S}, \cT, E_Q, Q ) & \geq \log ( 1 - 2\alpha) + n - |\Phi| \\
        & \geq n - |\Phi| - 10 \alpha,
    \end{align*}
    where the last lower bound follows from $\alpha \in (0, 0.1)$, and $\log ( 1 - 2 \alpha) \geq -10 \alpha$ in this interval.

        Combining the two bounds, we obtain
    \begin{equation}\label{eq:x_intermediate}
        \widetilde{H}_\infty ( X | S_{<Q}, \cT, E_Q {=} 1, Q ) \geq n - |\Phi| - 10 \alpha.
    \end{equation}
    Again applying \pref{lem:dors} to \pref{eq:x_intermediate}, we get
    \begin{equation*}
        \widetilde{H}_\infty ( X | S_{<Q}, \cT, B_Q, E_Q {=} 1, Q ) \geq n - |\Phi| - 10 \alpha - 1,
    \end{equation*}
    which completes the proof of the lemma.
\end{proof}

\subsubsection{Small Correlation between \texorpdfstring{$S_Q$}{S\_Q} and \texorpdfstring{$X$}{X}}

Finally, we show that $S_Q$ and $X$ are weakly correlated after conditioning on $Z$ and $E_Q {=} 1$. We show the following lemma.
\begin{lemma} \label{lem:correlation}
    \begin{equation*}
        I ( S_Q ; X | Z, E_Q {=} 1 ) \leq \frac{1}{(1-\alpha)^2} \frac{|U|}{m}
    \end{equation*}
\end{lemma}
\begin{proof}
        We start by substituting our choice of $Z$ into $I(S_Q; X | Z, E_Q {=} 1)$:
    \begin{align*}
        I ( S_Q ; X | Z, E_Q {=} 1 ) & \leq  I ( S_Q ; X U | S_{<Q}, \cT, B_Q, E_Q {=} 1, Q ) \\
        & \leq I ( S_Q ; X U | S_{<Q}, \cT, E_Q {=} 1, Q)
    \end{align*}
    where the last inequality follows from \pref{fact:chainrule2} with
    \begin{equation*}
        I ( S_Q ; B_Q | X U, S_{<Q}, \cT, E_Q {=} 1, Q ) = 0,
    \end{equation*}
    as $U, A_{\cQ_{all}}$ and Bob's independent randomness fully determine $B_Q$.
    That is, depending on $U$, either the final output of $A_Q$ is chosen as $B_Q$, or $B_Q$ is uniformly random in $\{ \pm 1\}$.
    Then we can further upper bound the term of interest as
    \begin{align*}
        I ( S_Q ; X U | S_{<Q}, \cT, E_Q {=} 1, Q) & = I ( S_Q ; X U | S_{<Q}, \cT, E_Q^1 {=} 1, E_Q^2 {=} 1, Q) \\
        & \leq \frac{I ( S_Q ; X U | S_{<Q}, \cT, E_Q^1 {=} 1, E_Q^2, Q)}{ \Pr[ E^2_Q {=} 1 | S_{<Q}, \cT, E_Q^1 {=} 1, Q ] } \\
        & \leq \frac{I ( S_Q; X U | S_{<Q}, \cT, E^1_Q {=} 1, Q )}{ 1 - \alpha }
    \end{align*}
    where the last inequality follows from \pref{fact:chainrule2} with
    \begin{align*}
        I ( E^2_Q ; X U | S_Q, S_{<Q}, \cT, E^1_Q {=} 1, Q ) & \leq H ( E^2_Q | S_Q, S_{<Q}, \cT, E^1_Q {=} 1, Q ) = 0
    \end{align*}
    from \pref{lem:min_entropy_s_i}. Now $E^1_Q$ is fully determined by $S_{<Q}, \cT, Q$, and is $1$ with probability at least $ 1 - \alpha$. Therefore
    \begin{equation*}
        I ( S_Q; X U | S_{<Q}, \cT, E^1_Q {=} 1, Q ) \leq \frac{I ( S_Q ; X U | S_{<Q}, \cT, Q )}{1 - \alpha}.
    \end{equation*}

        Finally, using the standard direct-sum technique and the chain rule for mutual information (\pref{fact:chainrule}), we get
    \begin{align*}
        I ( S_Q ; X U | S_{<Q}, \cT, Q ) & = \frac{1}{m} \sum_{q=1}^m I ( S_q ; X U | S_{< q}, \cT ) = \frac{ I ( \vec{S} ; X U | \cT ) }{m} \\
        & \leq \frac{I ( \vec{S} ; X U | \Phi ) }{m} \leq \frac{I ( \vec{S} ; X U ) }{m} \\
        & = \frac{ I ( \vec{S} ; X ) + I ( \vec{S} ; U | X) }{m} \leq \frac{|U|}{m}
    \end{align*}
    where the first inequality follows from \pref{fact:chainrule2} with $I ( A_{\cQ_{all}} ; X U | \vec{S}, \Phi ) = 0$. This is true as $\vec{S}, \Phi$ and Alice's independent randomness determine Alice's message $A_{\cQ_{all}}$. The second inequality holds again from \pref{fact:chainrule2} with $I ( \Phi ; \vec{S} | X U ) = 0$ as $\Phi$ is fully determined by $U$, $X$ and Bob's independent randomness. The final inequality holds from $I ( \vec{S}; X) = 0$ and $I ( \vec{S} ; U | X) \leq \abs{U}$.

        Combining the inequalities, we obtain the desired bound
    \begin{equation*}
        I ( S_Q; X | Z, E_Q {=} 1 ) \leq \frac{1}{(1-\alpha)^2} \frac{|U|}{m}.
    \end{equation*}
\end{proof}

\subsubsection{The Main Contradiction}

We will use the following combinatorial lemma from \cite{ko_lower_2025}. In its statement, and in its proof in \pref{app:missing}, we abbreviate the conditioning on a fixed pair $(a,z)$ as $\cE := ( A = a, Z = z )$.

\begin{lemma}[\cite{ko_lower_2025}] \label{lem:comb}
        Let $X$ and $Y$ be distributions over $\{0,1\}^n$. For any $\gamma \geq 3$, if a random process $Z = z$ and $A$, which contains $z_{out}$, satisfy the inequalities
    \begin{align*}
    & \widetilde{H}_\infty ( Y | A, Z = z ) + \widetilde{H}_\infty ( X | A, Z = z ) \geq n + 2 \gamma \\
    & I ( Y ; X | A, Z = z ) \leq 2^{-2 \gamma},
    \end{align*}
    then it must be the case that
    \begin{equation*}
        \E_{A | Z = z } \left[ \abs{  \E_{Y, X |_{\cE}} \left[ \chi_{Y} (X) \cdot z_{out} \right] } \right] < 2^{- \gamma + 2}
    \end{equation*}
\end{lemma}

We include the full proof of \pref{lem:comb} in \pref{app:missing} for completeness. The following lemma completes the proof of \pref{thm:main_informal}.

\begin{lemma} \label{lem:comm_lb}
If $m = n^{1+ \Omega(1)}$, then no 2.5-round Multiphase Communication Game satisfying the parameters of \pref{cor:choice_parameters} can exist.
\end{lemma}

To see how \pref{lem:comm_lb} completes the proof of \pref{thm:main_informal}, note that if $t_{tot} = o ( ( \log n / \log ( t_u w ) )^2 )$, then \pref{cor:choice_parameters} yields a 2.5-round protocol with the stated parameters, whereas \pref{lem:comm_lb} shows that no such protocol exists. This gives the claimed lower bound on $t_{tot}$.

\begin{proofof}{\pref{lem:comm_lb}}
        Recall that \pref{cor:choice_parameters} implies \pref{cl:advantage}, i.e.,
    \begin{equation*}
        \E_{Z} \left[ \abs{ \E_{(S_Q,X)|Z} \left[ \chi_{S_Q} (X) \right] } \right]  \geq n^{- o(1)}
    \end{equation*}
    with the parameters
    \begin{align*}
                & \abs{U} \leq n^{1 + o(1)} ,~\abs{\Phi} \leq n - \frac{n^{2/3}}{4} , \\
        & \abs{A_q} = \Theta ( w \cdot t_{tot} ) \leq o ( \log^3 (n) ).
    \end{align*}

    On the other hand, from \pref{lem:min_entropy_s_i}, \pref{lem:min_entropy_x}, and \pref{lem:correlation}, we get the following set of inequalities:
    \begin{align}
        \widetilde{H}_\infty ( X | Z, E_Q {=} 1) & \geq n - |\Phi| - 10 \alpha - 1 \label{eq:x_bound_final} \\
        \widetilde{H}_\infty ( S_Q | Z, E_Q {=} 1 ) & \geq n - \frac{8 w \cdot t_{tot}}{\alpha^2} - 1 \label{eq:s_bound_final} \\
        & = n - \Theta \left( \frac{\abs{A_q}}{\alpha^2} \right) - 1 \nonumber \\
        I ( S_Q ; X | Z, E_Q {=} 1 ) & \leq \frac{1}{(1-\alpha)^2} \frac{|U|}{m} \label{eq:correlation_final}
    \end{align}
    where \pref{eq:s_bound_final} simply follows from
    \begin{align*}
        \widetilde{H}_\infty ( S_Q | Z, E_Q {=} 1 ) & = \widetilde{H}_\infty ( S_Q | B_Q, \cT, S_{< Q}, Q, E_Q {=} 1 ) \\
        & \geq \widetilde{H}_\infty ( S_Q | \cT, S_{< Q}, Q, E_Q {=} 1 ) - 1 \\
        & \geq n - \frac{8 w \cdot t_{tot}}{\alpha^2} - 1
    \end{align*}
    where the first inequality follows from \pref{lem:dors} and the second inequality follows from \pref{lem:min_entropy_s_i}.

    Adding \pref{eq:x_bound_final} and \pref{eq:s_bound_final}, and plugging in the bounds for $\abs{U}, \abs{\Phi}, \abs{A_q}$, we obtain
    \begin{align*}
        \widetilde{H}_\infty ( X | Z, E_Q {=} 1) + \widetilde{H}_\infty ( S_Q | Z, E_Q {=} 1 ) & \geq \frac{n^{2/3}}{4} + \left( n - \frac{\log^3 (n)}{\alpha^2} \right) \\
        I ( S_Q; X | Z, E_Q {=} 1 ) & \leq O \left( \frac{1}{(1-\alpha)^2} \frac{n^{1 + o(1)}}{m} \right)
    \end{align*}
    upon which we can invoke \pref{lem:comb}. By setting $\alpha := n^{-1/5}$, and $m := n^{ 1 + \Omega(1)}$,
    our $\gamma$ for \pref{lem:comb} becomes $\gamma := \Theta ( \log n )$. In particular, we get
    \begin{align*}
        & \widetilde{H}_\infty ( X | Z, E_Q {=} 1) + \widetilde{H}_\infty ( S_Q | Z, E_Q {=} 1 ) \geq n + n^{\Theta(1)} \\
        &  I ( S_Q; X | Z, E_Q {=} 1 ) \leq 2^{-\Theta ( \log n )}
    \end{align*}
    Here, we remark that \pref{eq:correlation_final} is the bottleneck for $\gamma$. This then implies
    \begin{equation*}
        \E_{Z, E_Q {=} 1} \left[ \abs{ \E_{(S_Q,X)|Z, E_Q {=} 1} \left[ \chi_{S_Q} (X) \right] } \right] < n^{-\Omega(1)}.
    \end{equation*}
    However, this contradicts \pref{cl:advantage}, due to the following chain of inequalities.
    \begin{align}
        & n^{-o(1)} \leq \E_{Z} \left[ \abs{ \E_{(S_Q,X)|Z} \left[ \chi_{S_Q} (X) \right] } \right] \nonumber \\
        & \leq \E_{Z} \left[ \Pr [ E_Q {=} 1 | Z = z ] \cdot \abs{ \E_{(S_Q,X)|Z = z, E_Q {=} 1} \left[ \chi_{S_Q} (X) \right] } \right] \nonumber \\
        & \; + \underbrace{\E_{Z}\left[ \Pr [ E_Q {=} 0 | Z = z ] \cdot \abs{ \E_{(S_Q,X)|Z = z, E_Q {=} 0}  \left[ \chi_{S_Q} (X) \right] }\right]}_{\leq~ \Pr[ E_Q {=} 0 ]} \nonumber \\
        & \leq \E_{Z} \left[ \abs{ \E_{(S_Q,X)|Z, E_Q {=} 1} \left[ \chi_{S_Q} (X) \right] } \right] + \Pr[ E_Q {=} 0 ] \nonumber \\
        & \leq \E_{Z, E_Q {=} 1} \left[ \abs{ \E_{(S_Q,X)|Z, E_Q {=} 1} \left[ \chi_{S_Q} (X) \right] } \right] \nonumber \\
        & \quad + \norm{ Z - Z |_{E_Q {=} 1} }_1 + \Pr[ E_Q {=} 0 ] \label{eq:final_form_advantage}
    \end{align}
        Here the underbraced term is at most $\Pr[E_Q {=} 0]$ since $\abs{\chi_{S_Q} (X)} \leq 1$, and the last inequality follows from the standard fact $\abs{ \E_{X} [ f(X) ] - \E_{Y} [ f(Y) ] } \leq \norm{f}_\infty \cdot \norm{X-Y}_1$.
    Write $\pi := \Pr[ E_Q {=} 1 ]$, $Z_1 := Z |_{E_Q {=} 1}$ and $Z_0 := Z |_{E_Q {=} 0}$, so that $Z = \pi Z_1 + ( 1 - \pi ) Z_0$. Then $\norm{ Z - Z_1 }_1$ can be upper bounded by
    \begin{align}
        \norm{ Z - Z_1 }_1 & = \norm{ \pi Z_1 + ( 1 - \pi ) Z_0 - \pi Z_1 - ( 1 - \pi ) Z_1 }_1 \nonumber \\
        & \leq ( 1 - \pi ) \cdot \norm{ Z_0 - Z_1 }_1 \nonumber \\
        & \leq 2 \cdot ( 1 - \pi ) = 2 \cdot \Pr [E_Q {=} 0 ]. \label{eq:simple_l1_bound}
    \end{align}
        Combining \pref{eq:final_form_advantage} with \pref{eq:simple_l1_bound} yields the final contradiction:
    \begin{align}
         n^{-o(1)}  & \leq \pref{eq:final_form_advantage} \nonumber \\
         & \leq \E_{Z, E_Q {=} 1} \left[ \abs{ \E_{(S_Q,X)|Z, E_Q {=} 1} \left[ \chi_{S_Q} (X) \right] } \right] \nonumber \\
         & \quad + 3 \cdot \Pr[ E_Q {=} 0 ] \nonumber \\
         &  \leq n^{-\Omega(1)} + O( \alpha ) \leq n^{-\Omega(1)} \label{eq:contradiction}
    \end{align}
    The final inequality on $\Pr[ E_Q {=} 0 ] = 1 - \Pr[ E_Q {=} 1 ]$ follows from \pref{lem:min_entropy_s_i} as
    \begin{equation*}
        \Pr[ E_Q {=} 1 ] \geq ( 1 - \alpha )^2 \geq 1 - 2 \alpha.
    \end{equation*}
        With $\alpha := n^{-1/5}$, \pref{eq:contradiction} states $n^{-o(1)} \leq n^{-\Omega(1)}$, a contradiction.

\end{proofof}

\section*{Acknowledgments}
The author thanks Huacheng Yu and anonymous reviewers for carefully reading an earlier draft of this paper and providing valuable feedback.

\bibliographystyle{alphaurl}
\bibliography{references}

\appendix
    \section{Omitted Proofs of Claims and Lemmas} \label{app:missing}

    To make the presentation self-contained, we include the omitted proofs of claims and lemmas.

    \begin{repclaim}{cl:conditioning_infty}
        \begin{equation*}
            \widetilde{H}_\infty ( A | B,C ) \leq \widetilde{H}_\infty ( A | B  )
        \end{equation*}
    \end{repclaim}
\begin{proof}
        We first show the following inequality,
    \begin{align} \label{eq:conditioning_infty}
        & \E_{c \sim C |_{B = b} } \left[ \max_{a} \Pr[ A = a | B = b, C =c ] \right] \nonumber \\
                & \qquad \geq \max_{a} \Pr[ A = a | B = b ],
    \end{align}
    which is a standard fact about the $\max$ function. Let $a^* := \argmax_{a} \Pr[ A = a | B = b ] $. Then
    \begin{align*}
        \Pr[ A = a^* | B = b ] = \E_{c \sim C |_{B = b} } \left[ \Pr[ A = a^* | B = b, C =c ] \right] & \leq \E_{c \sim C |_{B = b} } \left[ \max_{a} \Pr[ A = a | B = b, C =c ] \right]
    \end{align*}

    With \pref{eq:conditioning_infty} established, we are ready to prove the claim. Recall that
    \begin{align*}
        \widetilde{H}_\infty ( A | B,C ) & := - \log \left( \E_{b,c \sim B,C } \left[ \max_{a} \Pr[ A = a | B = b, C =c ] \right] \right)
    \end{align*}
    Taking expectation over $B$ on both sides of \pref{eq:conditioning_infty}, we have
    \begin{align*}
        \E_{b,c \sim B,C } \left[ \max_{a} \Pr[ A = a | B = b, C =c ] \right] & \geq \E_{b \sim B } \left[ \max_{a} \Pr[ A = a | B = b ] \right]
    \end{align*}
    Therefore, we get
    \begin{align*}
        \widetilde{H}_\infty ( A | B,C ) & = - \log \left( \E_{b,c \sim B,C } \left[ \max_{a} \Pr[ A = a | B = b, C =c ] \right] \right) \\
        & \leq - \log \left( \E_{b \sim B } \left[ \max_{a} \Pr[ A = a | B = b ] \right] \right) \\
        & = \widetilde{H}_\infty ( A | B )
    \end{align*}

\end{proof}

    \begin{repclaim}{cl:lowL2}[\cite{jain_property_2009}] 
        Let $\cD$ be a distribution and $\cU$ the uniform distribution over some $\cS \subseteq \{0,1\}^n$. If $D( \cD || \cU ) \leq t$ with $t > 1$, then for every $\alpha \geq 4$ there exists an event $E$ such that
    \begin{align*}
        & \cD ( E ) \geq 1 - \frac{1}{\alpha} \\
        & H_\infty ( \cD |_{E} ) \geq  \log |\cS| - 2 \alpha t
    \end{align*}
    \end{repclaim}
\begin{proof}
    We partition the domain according to $\cD (x)$. Let $E$ denote the set of $x$ such that
    \begin{equation*}
        \log \frac{\cD(x)}{\cU (x)} < 2 \alpha t - 1,
    \end{equation*}
    and let $p := 1 - \cD (E)$ and $q := 1 - \cU (E)$, so that $q \leq 2^{- ( 2 \alpha t - 1 )} \, p$. By the data-processing inequality (or log-sum inequality) for the indicator of $E$,
    \begin{align*}
        t \geq D ( \cD || \cU ) & \geq p \log \frac{p}{q} + ( 1 - p ) \log \frac{1 - p}{1 - q} \\
        & \geq p \, ( 2 \alpha t - 1 ) - p \log e,
    \end{align*}
    where the last inequality uses $1 - q \leq 1$ and $( 1 - p ) \log ( 1 - p ) \geq - p \log e$. As $\alpha t \geq 4 \geq 1 + \log e$, we get $p \leq t / ( 2 \alpha t - 1 - \log e ) \leq t / ( \alpha t ) = 1 / \alpha$. In particular $\cD (E) \geq 1 - 1/\alpha \geq 1/2$, so for any $x$ in the support of $\cD|_{E}$,
    \begin{equation*}
        \cD|_{E} (x) = \frac{\cD (x)}{\cD (E)} \leq 2 \cdot \cD (x) < 2^{2 \alpha t} \cdot \cU (x) = 2^{- \log |\cS| + 2 \alpha t},
    \end{equation*}
    which then gives $H_\infty ( \cD |_{E} ) \geq \log |\cS| - 2 \alpha t$.
\end{proof}

Towards the proof of \pref{lem:comb}, we need the following well-known fact about the Hadamard matrix.

\begin{claim}[Lindsey's Lemma] \label{cl:hadamard}
Let $H$ be a $2^n \times 2^n$ Hadamard matrix. Let $\mu$ and $\nu$ be distributions over $\{0,1\}^n$, written as vectors. Then
\begin{equation*}
    \abs{ \mu^{T} H \nu } \leq \norm{ \mu }_2 \norm{ \nu }_2 \cdot 2^{n/2}
\end{equation*}
\end{claim}

Also, \pref{fact:renyi} implies the following simple proposition about the $\ell_2^2$ norm of the distribution versus its min-entropy.

\begin{proposition} \label{prop:l-2-min}
    \begin{equation*}
        \E_{b \sim B} \left[ \sum_{a} \Pr [ A = a | B = b]^2 \right] \leq 2^{- \widetilde{H}_\infty (A | B )}
    \end{equation*}
\end{proposition}
\begin{proof}
    \begin{equation*}
        \sum_{a} \Pr [ A = a | B = b]^2 = 2^{ - H_2 (A | B = b ) } \leq 2^{ - H_\infty (A | B = b ) }
    \end{equation*}
    Therefore,
    \begin{align*}
        \E_{b \sim B} \left[ \sum_{a} \Pr [ A = a | B = b]^2 \right] & \leq \E_{b \sim B} \left[ 2^{- H_\infty (A | B = b )} \right] = 2^{ - \widetilde{H}_\infty (A | B ) }
    \end{align*}
    where the last equality follows from the definition of average min-entropy.
\end{proof}

Now we are ready to prove \pref{lem:comb}. We prove an equivalent statement. Recall that we abbreviate the conditioning on a fixed pair $(a,z)$ as $\cE := ( A = a, Z = z )$.
\begin{replemma}{lem:comb}
    No setting of a random variable $Z = z$ and $A$, which contains $z_{out}$, can simultaneously satisfy all three of the following inequalities for any $\gamma \geq 3$
    \begin{align}
    & \E_{A | Z = z } \left[ \abs{  \E_{Y, X |_{\cE}} \left[ \chi_{Y} (X) \cdot z_{out} \right] } \right] \geq 2^{- \gamma + 2} \label{eq:condition1} \\
    & \widetilde{H}_\infty ( Y | A, Z = z ) + \widetilde{H}_\infty ( X | A, Z = z ) \geq n + 2 \gamma  \label{eq:condition2} \\
    & I ( Y; X | A, Z = z ) \leq 2^{-2 \gamma} \label{eq:condition3}
    \end{align}
\end{replemma}

    \begin{proof}
    For the sake of contradiction, suppose such $z$ and $A$ exist. First, as $z_{out}$ is $\pm 1$, the absolute value of the guess does not change. That is
    \begin{equation}
        \abs{  \E_{Y, X |_{\cE}} \left[ \chi_{Y} (X) \cdot z_{out} \right] } =  \abs{ \E_{ Y, X |_{\cE}} \left[ \chi_{Y} (X) \right] }
    \end{equation}
    Then we can use the $\ell_1$ bound to have
    \begin{align}
        \abs{ \E_{ Y, X |_{\cE}} \left[ \chi_{Y} (X) \right] } & \leq \abs{ Y |_{\cE} \cdot H \cdot X |_{\cE} } \label{eq:invariant_hadamard} \\
        & +  \norm{ Y |_{\cE} \times X |_{\cE} - (Y,X) |_{\cE} }_1  \label{eq:invariant_rectangle}
    \end{align}

    We bound the expectation of \pref{eq:invariant_rectangle}. Our KL-divergence term is then equal to the mutual information between $Y$ and $X$ conditioned on $Z = z$. Namely, using the chain rule for the KL-divergence,
    \begin{align*}
        & D ( (Y, X) |_{\cE}  || Y |_{\cE} \times X |_{ \cE } ) \\
        & = \underbrace{ D( X |_{\cE} || X |_{\cE} )}_{ = 0} + \E_{x \sim X |_{\cE}} \left[ D ( Y|_{X = x, \cE}  || Y |_{\cE} )  \right] \\
        & = I ( Y; X | \cE )
    \end{align*}
    Then, due to Pinsker's inequality (\pref{fact:pinsker}), we have
    \begin{align*}
        & \norm{ Y |_{\cE} \times X |_{\cE} - (Y,X) |_{\cE} }_1 \leq 2 \sqrt{ I ( Y ; X | \cE ) }
    \end{align*}
    Taking expectation over $A$ and applying Jensen's inequality,
    \begin{align}
        \E_{A |_{Z = z}} \left[ \norm{ Y |_{\cE} \times X |_{\cE} - (Y,X) |_{\cE} }_1 \right] & \leq 2 \sqrt{ I ( Y ; X | A, Z = z ) } \leq 2^{-\gamma + 1} \label{eq:invariant_rectangle_expectation}
    \end{align}
    where the last bound follows from \pref{eq:condition3}.

    Next, we bound \pref{eq:invariant_hadamard}. Due to \pref{cl:hadamard} and the Cauchy--Schwarz inequality,
    \begin{align*}
        & \E_{A |_{Z = z}} \left[ \abs{ Y |_{\cE} \cdot H \cdot X |_{\cE} } \right] \leq \E_{A |_{Z = z}} \left[ 2^{n/2} \cdot \norm{ Y |_{\cE} }_2 \cdot \norm{X |_{\cE} }_2 \right] \\
        & \leq 2^{n/2} \cdot \sqrt { \E_{A |_{Z = z}} \left[ \norm{ Y |_{\cE} }_2^2 \right] \cdot \E_{A |_{Z = z}} \left[ \norm{ X |_{\cE} }_2^2 \right] }
    \end{align*}
    \pref{prop:l-2-min} implies
    \begin{align*}
        & \E_{A |_{Z = z}} \left[ \norm{ Y |_{\cE} }_2^2 \right] \leq 2^{ - \widetilde{H}_\infty ( Y | A, Z =z ) } \\
        & \E_{A |_{Z = z}} \left[ \norm{ X |_{\cE} }_2^2 \right] \leq 2^{ - \widetilde{H}_\infty ( X | A, Z =z ) }
    \end{align*}
    which would in turn imply
    \begin{align*}
        \sqrt { \E_{A |_{Z = z}} \left[ \norm{ Y |_{\cE} }_2^2 \right] \cdot \E_{A |_{Z = z}} \left[ \norm{ X |_{\cE} }_2^2 \right] } & \leq 2^{ - ( \widetilde{H}_\infty ( Y | A, Z =z ) + \widetilde{H}_\infty ( X | A, Z =z ) ) / 2 } \\
        & \leq 2^{- \frac{n + 2 \gamma}{2}} = 2^{-n/2} \cdot 2^{ -\gamma}
    \end{align*}
    which then implies \pref{eq:invariant_hadamard} is upper bounded by
    \begin{equation}
        \pref{eq:invariant_hadamard} \leq 2^{n/2} \cdot 2^{-n/2} \cdot 2^{-\gamma} = 2^{- \gamma}
    \end{equation}

    Therefore, we get
    \begin{align}
    \E_{A | Z = z } \left[ \abs{  \E_{Y, X |_{\cE}} \left[ \chi_{Y} ( X ) \cdot z_{out} \right] } \right] \leq \frac{3}{2^{\gamma}}
    \end{align}
    which contradicts \pref{eq:condition1}.
\end{proof}

\section{Lifting Theorem} \label{app:lifting}

We consider the following ``hard'' functions, which, roughly speaking, are functions $f$ with small discrepancy under every product distribution of large min-entropy.
\begin{definition}[Hard Functions] \label{def:hard}
    Let $\cS \subseteq \{0,1\}^n$. We say $f : \{0,1\}^{2n} \to \{ \pm 1 \}$ is hard (with respect to $\cS$) if there is a constant $c > 0$ such that, for all large enough $n$, the following holds for all random variables $S$ over $\cS$, $X$ over $\{0,1\}^n$, and $A$ such that $S$ and $X$ are independent conditioned on $A = a$, for every $a$. If
    \begin{align*}
        & \widetilde{H}_\infty ( S | A ) \geq \log \abs{\cS} - \log^3 n, \qquad \widetilde{H}_\infty ( X | A ) \geq \frac{\sqrt{n}}{3},
    \end{align*}
    then
    \begin{equation*}
        \E_{a \sim A} \left[ \abs{ \E_{(S,X) | A = a} \left[ f(S,X) \right] } \right] \leq n^{-c}.
    \end{equation*}
\end{definition}

\paragraph{Example: Inner Product} Inner Product over $\Field_2$ satisfies the definition with $\cS = \{0,1\}^n$ and any constant $c > 0$. By \pref{cl:hadamard} together with $\norm{\mu}_2^2 \leq 2^{-H_\infty ( \mu )}$ (\pref{fact:renyi}), for every $a$,
\begin{equation*}
    \abs{ \E_{(S,X) | A = a} [ f(S,X) ] } \leq 2^{n/2} \cdot 2^{- H_\infty ( S | A = a ) / 2} \cdot 2^{- H_\infty ( X | A = a ) / 2}.
\end{equation*}
Averaging over $a$ and applying the Cauchy--Schwarz inequality,
\begin{align*}
    \E_{a \sim A} \left[ \abs{ \E_{(S,X) | A = a} [ f(S,X) ] } \right] & \leq 2^{n/2} \left( \E_{a \sim A} \left[ 2^{- H_\infty ( S | A = a )} \right] \right)^{1/2} \left( \E_{a \sim A} \left[ 2^{- H_\infty ( X | A = a )} \right] \right)^{1/2} \\
    & = 2^{( n - \widetilde{H}_\infty ( S | A ) - \widetilde{H}_\infty ( X | A ) ) / 2} \leq 2^{( \log^3 n - \sqrt{n}/3 ) / 2},
\end{align*}
which is $2^{-\Omega(\sqrt{n})}$. Disjointness, on the other hand, does not satisfy the definition: sampling and revealing a few coordinates of $X$ already gives a nontrivial advantage (see \cite{braverman_information_2013}).

\paragraph{A Good Event for Small Divergence} The proof of \pref{cl:lowL2} also gives the following corollary for the regime in which the divergence is much smaller than $1$.
\begin{corollary}[of the proof of \pref{cl:lowL2}] \label{cor:small_divergence}
    Let $S$ be uniform over a finite set $\cS$, and let $Y$ be a random variable with $I ( S ; Y ) \leq \eps \leq 1/8$. Let $E$ be the event that the chosen pair $(s, y)$ of $(S, Y)$ satisfies $\Pr [ S = s | Y = y ] \leq 2 / \abs{\cS}$. Then
    \begin{equation*}
        \Pr [ E ] \geq 1 - \sqrt{2 \eps} \geq \frac{1}{2}, \qquad \widetilde{H}_\infty ( S | Y, E {=} 1 ) \geq \log \abs{\cS} - 2.
    \end{equation*}
\end{corollary}
\begin{proof}
    Denote by $\mu_y$ the distribution of $S$ conditioned on $Y = y$, and by $\cU$ the uniform distribution over $\cS$, which is the original distribution of $S$. We follow the proof of \pref{cl:lowL2} for each $\cD = \mu_y$, with threshold $1$: conditioned on $Y = y$, the event $E$ is the set
    \begin{equation*}
        E_y := \left\{ s : \log \frac{\mu_y (s)}{\cU (s)} \leq 1 \right\} = \left\{ s : \mu_y (s) \leq \frac{2}{\abs{\cS}} \right\}.
    \end{equation*}
    The main change is the bound on $\Pr [ E ]$, for which, as the divergence is now below $1$, we use Pinsker's inequality. On the complement $\cA_y := \cS \setminus E_y$, we have $\mu_y > 2 \cU$, so that
    \begin{equation*}
        \mu_y ( \cA_y ) \leq 2 \left( \mu_y ( \cA_y ) - \cU ( \cA_y ) \right) \leq 2 \norm{ \mu_y - \cU }_{TV}.
    \end{equation*}
    Averaging over $y$, and using Pinsker's inequality (\pref{fact:pinsker}), Jensen's inequality and $\log e \geq 1$,
    \begin{align*}
        \Pr [ E {=} 0 ] = \E_y \left[ \mu_y ( \cA_y ) \right] & \leq 2 \E_y \left[ \sqrt{ \frac{ D ( \mu_y || \cU ) }{2 \log e} } \right] \\
        & \leq \sqrt{ 2 \E_y \left[ D ( \mu_y || \cU ) \right] } = \sqrt{ 2 I ( S ; Y ) } \leq \sqrt{2 \eps} \leq \frac{1}{2}.
    \end{align*}
    For the min-entropy, as in the proof of \pref{cl:lowL2}, for every $y$ with $\mu_y ( E_y ) > 0$ and every $s \in E_y$,
    \begin{equation*}
        \Pr [ S = s | Y = y, E {=} 1 ] = \frac{\mu_y (s)}{\mu_y ( E_y )} \leq \frac{2}{\abs{\cS} \, \mu_y ( E_y )}.
    \end{equation*}
    Conditioned on $E {=} 1$, the value $y$ has probability $\Pr [ Y = y ] \, \mu_y ( E_y ) / \Pr [ E ]$, so averaging over $y$ gives
    \begin{equation*}
        2^{- \widetilde{H}_\infty ( S | Y, E {=} 1 )} \leq \sum_{y} \frac{ \Pr [ Y = y ] \, \mu_y ( E_y ) }{ \Pr [ E ] } \cdot \frac{2}{\abs{\cS} \, \mu_y ( E_y )} \leq \frac{2}{\abs{\cS} \Pr [ E ]} \leq \frac{4}{\abs{\cS}}. \qedhere
    \end{equation*}
\end{proof}

We are now ready to prove the following general lifting theorem, which extends \pref{thm:main_informal} to every hard function.

\begin{theorem}[Lifting Theorem] \label{thm:lifting}
    Let $f$ be hard (\pref{def:hard}) and let $m = n^{1+\Omega(1)}$. Then any data structure for the Multiphase Problem for $f$ whose update time satisfies $t_u = n^{o(1)}$ must have
    \begin{equation*}
        t_{tot} \geq \Omega \left( \left( \frac{\log n}{\log ( t_u w )} \right)^2 \right).
    \end{equation*}
\end{theorem}
\begin{proof}
    The proof follows that of \pref{lem:comm_lb}, except that $Z$ contains the transcripts of only a random set of $k \ll m$ queries, as in \cite{ko_adaptive_2020}, and \pref{def:hard} replaces \pref{lem:comb}, which is specific to Inner Product. The protocol does not know which queries are in the set, so the transcripts of the other queries in the set carry little information about $S_Q$ and \pref{cor:small_divergence} applies. The transcript of $Q$ itself costs only $\abs{A_q}$ bits of average min-entropy by \pref{lem:dors}.

    \paragraph{Setup} Suppose, for contradiction, that $t_{tot} = o ( ( \log n / \log ( t_u w ) )^2 )$, and let $\Pi$ be the 2.5-round protocol of \pref{cor:choice_parameters}, so that
    \begin{equation*}
        \abs{U} \leq n^{1 + o(1)} ,~\abs{\Phi} \leq n - \frac{n^{2/3}}{4} ,~\abs{A_q} \leq o ( \log^3 (n) ),
    \end{equation*}
    and $\adv_{\vec{S}} ( G_{f}, \Pi ) \geq n^{-o(1)}$ for every $\vec{S}$. As the hard distribution, we select each $S_q$ independently and uniformly from $\cS$, and $X$ uniformly at random. We now define the random variables used in the rest of the proof.
    \begin{itemize}
        \item \emph{The random set of queries.} Let $\sigma \in \mathrm{Sym}([m])$ be a uniformly random permutation and $j \in [k]$ a uniformly random index, independent of each other and of everything else, where
        \begin{equation*}
            k := \sqrt{nm}.
        \end{equation*}
        We need $k \gg \abs{U}$ for the correlation bound and $k \ll m / \abs{A_q}$ for the information bound below; as $m = n^{1 + \Omega(1)}$, this $k$ satisfies both, and $k \leq m/2$ for large enough $n$. The random set of queries is $\sigma([k])$, and we write
        \begin{equation*}
            Q := \sigma(j), \qquad J' := \sigma ( [k] \setminus \{ j \} )
        \end{equation*}
        for the query we analyze and for the other queries in the set, respectively.
        \item \emph{The conditioning.} Let
        \begin{align*}
            Z & := \left( \cT_\sigma, S_{\sigma(<j)}, B_Q, \sigma, j \right), & \cT_\sigma & := \left( \Phi, A_{\sigma(1)}, \ldots, A_{\sigma(k)} \right), \\
            Y & := \left( \Phi, A_{J'}, S_{\sigma(<j)}, \sigma, j \right), & S_{\sigma(<j)} & := \left( S_{\sigma(1)}, \ldots, S_{\sigma(j-1)} \right).
        \end{align*}
        That is, $Z$ is the $Z$ of \pref{lem:comm_lb} with the transcripts restricted to the random set $\sigma([k])$, ordered by $\sigma$, and $Y$ is $Z$ without the messages $A_Q$ and $B_Q$ of $Q$ itself.
        \item \emph{The good event.} Let $E$ be the event of \pref{cor:small_divergence} for $S = S_Q$ and $Y$, that is, the event that the chosen pair $(s, y)$ of $(S_Q, Y)$ satisfies $\Pr [ S_Q = s | Y = y ] \leq 2 / \abs{\cS}$. As in \pref{sec:comm_lb}, we identify $E$ with its indicator random variable.
    \end{itemize}
    The permutation plays two roles: $Q$ is a uniformly random query, and $\sigma([k])$ is a uniformly random set of queries containing $Q$ that is independent of the protocol.

    \paragraph{Information from the Other Queries} We first show that the transcripts of the other queries in the random set reveal little information about $S_Q$.
    \begin{claim} \label{cl:other_queries}
        \begin{equation} \label{eq:other_queries}
            I ( S_Q ; Y ) \leq \frac{ \abs{\Phi} + k \abs{A_q} }{ m - k + 1 } \leq \frac{2 n}{m} + 2 \sqrt{\frac{n}{m}} \cdot \abs{A_q} = n^{-\Omega(1)}.
        \end{equation}
    \end{claim}
    \begin{proof}
        The argument is essentially from \cite{ko_adaptive_2020}, using the standard direct-sum technique. The intuition is as follows. Together, $\Phi$ and the transcripts $A_{J'}$ consist of at most $\abs{\Phi} + k \abs{A_q}$ bits, and they may reveal information about any of the $m$ coordinates of $\vec{S}$. However, the protocol does not know $\sigma$, so it cannot direct this information towards $S_Q$: given $J'$, the query $Q$ is uniform over the $m - k + 1$ queries outside $J'$, and as the coordinates of $\vec{S}$ are independent, the information splits among these queries, so that $S_Q$ receives only a $1 / ( m - k + 1 )$ fraction of it on average.

        Formally, as $S_Q$ is uniform over $\cS$ for every fixing of $( \sigma, j )$, and as $S_{\sigma(<j)}$ is part of $S_{J'}$, writing $Y_{\max} := ( \Phi, A_{J'}, S_{J'} )$,
        \begin{equation*}
            I ( S_Q ; Y ) \leq I ( S_Q ; Y_{\max} | \sigma, j ) = \E_{J'} \left[ \frac{1}{m - k + 1} \sum_{q \notin J'} I ( S_q ; Y_{\max} ) \right],
        \end{equation*}
        where the equality holds as, conditioned on $J'$, the query $Q$ is uniform over the $m - k + 1$ queries outside $J'$. For every fixed $J'$,
        \begin{equation*}
            \sum_{q \notin J'} I ( S_q ; Y_{\max} ) \leq I ( S_{[m] \setminus J'} ; Y_{\max} ) = I ( S_{[m] \setminus J'} ; \Phi, A_{J'} | S_{J'} ) \leq \abs{\Phi} + k \abs{A_q}.
        \end{equation*}
        Here the equality holds as $S_{[m] \setminus J'}$ is independent of $S_{J'}$, and the first inequality is the chain rule over the queries outside $J'$, together with
        \begin{equation*}
            I ( S_{q} ; Y_{\max} | S_{T} ) = I ( S_{q} ; Y_{\max}, S_{T} ) \geq I ( S_q ; Y_{\max} ) \qquad \text{for every set } T \not\ni q \text{ of queries},
        \end{equation*}
        as $S_q$ is independent of $S_T$. This proves the first inequality in \pref{eq:other_queries}, and the second one uses $\abs{\Phi} \leq n$, $k = \sqrt{nm}$ and $m - k \geq m/2$.
    \end{proof}

    \paragraph{Conditioning on a Good Event} By \pref{cl:other_queries}, \pref{cor:small_divergence} applies to $S = S_Q$ and $Y$ with $\eps := I ( S_Q ; Y ) \leq 1/8$, so that
    \begin{equation*}
        \Pr [ E ] \geq \frac{1}{2}, \qquad \eta := \Pr [ E {=} 0 ] \leq \sqrt{ 2 I ( S_Q ; Y ) } = n^{-\Omega(1)}.
    \end{equation*}
    Conditioned on $E {=} 1$, we obtain the analogues of \pref{lem:min_entropy_s_i}, \pref{lem:min_entropy_x} and \pref{lem:correlation}.
    \begin{claim} \label{cl:good_event}
        For large enough $n$,
        \begin{align}
        \widetilde{H}_\infty ( S_Q | Z, E {=} 1 ) & \geq \log \abs{\cS} - \log^3 n, \label{eq:s_lifting} \\
        \widetilde{H}_\infty ( X | Z, E {=} 1 ) & \geq \frac{\sqrt{n}}{3}, \label{eq:x_lifting} \\
        I ( S_Q ; X | Z, E {=} 1 ) & \leq \frac{ 2 \abs{U} }{ k } = n^{-\Omega(1)}. \label{eq:correlation_lifting}
        \end{align}
    \end{claim}
    \begin{proof}
        For \pref{eq:s_lifting}, recall that $Z$ consists of $Y$, $A_Q$ and $B_Q$, where $( A_Q, B_Q )$ takes at most $2^{\abs{A_q} + 1}$ values. So \pref{lem:dors}, applied conditioned on $E {=} 1$, and \pref{cor:small_divergence} give
        \begin{equation*}
        \widetilde{H}_\infty ( S_Q | Z, E {=} 1 ) \geq \widetilde{H}_\infty ( S_Q | Y, E {=} 1 ) - \abs{A_q} - 1 \geq \log \abs{\cS} - \abs{A_q} - 3 \geq \log \abs{\cS} - \log^3 n,
        \end{equation*}
        as $\abs{A_q} = o ( \log^3 n )$.

        For \pref{eq:x_lifting}, we follow the proof of \pref{lem:min_entropy_x}. Let $R_A$ denote Alice's private randomness and $W := ( \vec{S}, \Phi, R_A, \sigma, j )$. Then $W$ determines every $A_q$, hence $Y$ and $A_Q$, and hence also $E$, which is determined by $S_Q$ and $Y$. By the definition of average min-entropy (as in \pref{eq:x_min_infinity}) and $\Pr [ E ] \geq 1/2$, then \pref{cl:conditioning_infty}, and then \pref{lem:dors} for $B_Q$ and for $\Phi$,
        \begin{align*}
        \widetilde{H}_\infty ( X | Z, E {=} 1 ) & \geq \widetilde{H}_\infty ( X | Z, E ) - 1 \geq \widetilde{H}_\infty ( X | W, B_Q ) - 1 \geq \widetilde{H}_\infty ( X | \vec{S}, R_A, \sigma, j ) - \abs{\Phi} - 2 \\
        & = n - \abs{\Phi} - 2 \geq \frac{n^{2/3}}{4} - 2 \geq \frac{\sqrt{n}}{3}.
        \end{align*}

        For \pref{eq:correlation_lifting}, we follow the proof of \pref{lem:correlation} in three steps. First, we remove the conditioning on $E$, at the cost of a factor $1 / \Pr [ E ] \leq 2$:
        \begin{equation*}
            I ( S_Q ; X | Z, E {=} 1 ) \leq \frac{ I ( S_Q ; X U | Z, E ) }{ \Pr [ E ] } \leq 2 \, I ( S_Q ; X U | Z ).
        \end{equation*}
        The first inequality holds as $I ( S_Q ; X U | Z, E )$ is the average of $I ( S_Q ; X U | Z, E {=} b )$ over $b \in \{ 0, 1 \}$, weighted by $\Pr [ E {=} b ]$, and as adding $U$ can only increase the information. The second follows from \pref{fact:chainrule2} with $I ( X U ; E | S_Q, Z ) = 0$, as $E$ is determined by $S_Q$ and $Z$. Next, we drop $B_Q$ from $Z$:
        \begin{equation*}
            I ( S_Q ; X U | Z ) \leq I ( S_Q ; X U | \cT_\sigma, S_{\sigma(<j)}, \sigma, j ).
        \end{equation*}
        This again follows from \pref{fact:chainrule2}, now with $I ( S_Q ; B_Q | X U, \cT_\sigma, S_{\sigma(<j)}, \sigma, j ) = 0$, as $B_Q$ is determined by $U$, $A_Q$ and Bob's randomness, which is independent of $( \vec{S}, R_A, \sigma, j )$ given $( X, U, \Phi )$. Finally, as $j$ is uniform over $[k]$, the chain rule (\pref{fact:chainrule}) over the queries $\sigma(1), \ldots, \sigma(k)$, in this order, gives the direct-sum bound
        \begin{align*}
            I ( S_Q ; X U | \cT_\sigma, S_{\sigma(<j)}, \sigma, j ) & = \frac{1}{k} \E_{\sigma} \left[ \sum_{i = 1}^{k} I ( S_{\sigma(i)} ; X U | \cT_\sigma, S_{\sigma(<i)} ) \right] = \frac{1}{k} \E_{\sigma} \left[ I ( S_{\sigma([k])} ; X U | \cT_\sigma ) \right] \\
            & \leq \frac{1}{k} \E_{\sigma} \left[ I ( \vec{S} ; X U | \cT_\sigma ) \right] \leq \frac{ I ( \vec{S} ; X U | \Phi ) }{k} \leq \frac{\abs{U}}{k},
        \end{align*}
        where the first inequality follows from the chain rule, and the last two are as in the proof of \pref{lem:correlation}. Combining the three steps,
        \begin{equation*}
            I ( S_Q ; X | Z, E {=} 1 ) \leq \frac{2 \abs{U}}{k} \leq \frac{2 n^{1 + o(1)}}{\sqrt{nm}} = n^{-\Omega(1)}. \qedhere
        \end{equation*}
    \end{proof}

    \paragraph{Contradiction} We conclude the proof of the theorem by deriving a contradiction as in the proofs of \pref{cl:advantage} and \pref{lem:comm_lb}. As $Q$ is uniform over $[m]$ and independent of $\vec{S}$, $X$ and the randomness of $\Pi$,
    \begin{equation*}
        \E \left[ B_Q \cdot f ( S_Q, X ) \right] = \E_{\vec{S}} \left[ \adv_{\vec{S}} ( G_{f}, \Pi ) \right] \geq n^{-o(1)}.
    \end{equation*}
    As $\abs{f} \leq 1$ and $B_Q$ is determined by $Z$, restricting to $E$ loses at most $\eta$:
    \begin{equation} \label{eq:good_advantage_lifting}
        n^{-o(1)} \leq \E_{Z | E {=} 1} \left[ \abs{ \E_{(S_Q,X) | Z, E {=} 1} \left[ f ( S_Q, X ) \right] } \right] + O(\eta).
    \end{equation}
    For $z$ in the support of $Z$ conditioned on $E {=} 1$, consider the product of the two conditional marginals,
    \begin{equation*}
        R_z := S_Q |_{Z = z, E {=} 1} \times X |_{Z = z, E {=} 1}.
    \end{equation*}
    As $R_z$ has the same marginals, \pref{eq:s_lifting} and \pref{eq:x_lifting} let us apply \pref{def:hard} with $A$ distributed as $Z$ conditioned on $E {=} 1$ and $(S, X)$ distributed as $R_z$ given $A = z$, which gives
    \begin{equation*}
        \E_{Z | E {=} 1} \left[ \abs{ \E_{R_Z} [ f ] } \right] \leq n^{-c}.
    \end{equation*}
    On the other hand, as $\abs{f} \leq 1$, Pinsker's inequality (\pref{fact:pinsker}), Jensen's inequality and \pref{eq:correlation_lifting} give
    \begin{align*}
        \E_{Z | E {=} 1} \left[ \abs{ \E_{(S_Q,X) | Z, E {=} 1} \left[ f ( S_Q, X ) \right] - \E_{R_Z} [ f ] } \right] & \leq \E_{Z | E {=} 1} \left[ \norm{ (S_Q, X) |_{Z, E {=} 1} - R_Z }_1 \right] \\
        & \leq \sqrt{ 2 I ( S_Q ; X | Z, E {=} 1 ) } = n^{-\Omega(1)}.
    \end{align*}
    Combining the last two bounds with \pref{eq:good_advantage_lifting} and $\eta = n^{-\Omega(1)}$,
    \begin{equation*}
        n^{-o(1)} \leq n^{-c} + n^{-\Omega(1)},
    \end{equation*}
    which is a contradiction.
\end{proof}

\end{document}